\documentclass[runningheads]{llncs}
\usepackage{graphicx}
\usepackage[firstpage]{draftwatermark}
\SetWatermarkText{\hspace*{6in}\raisebox{8.8in}{\includegraphics[scale=0.15]{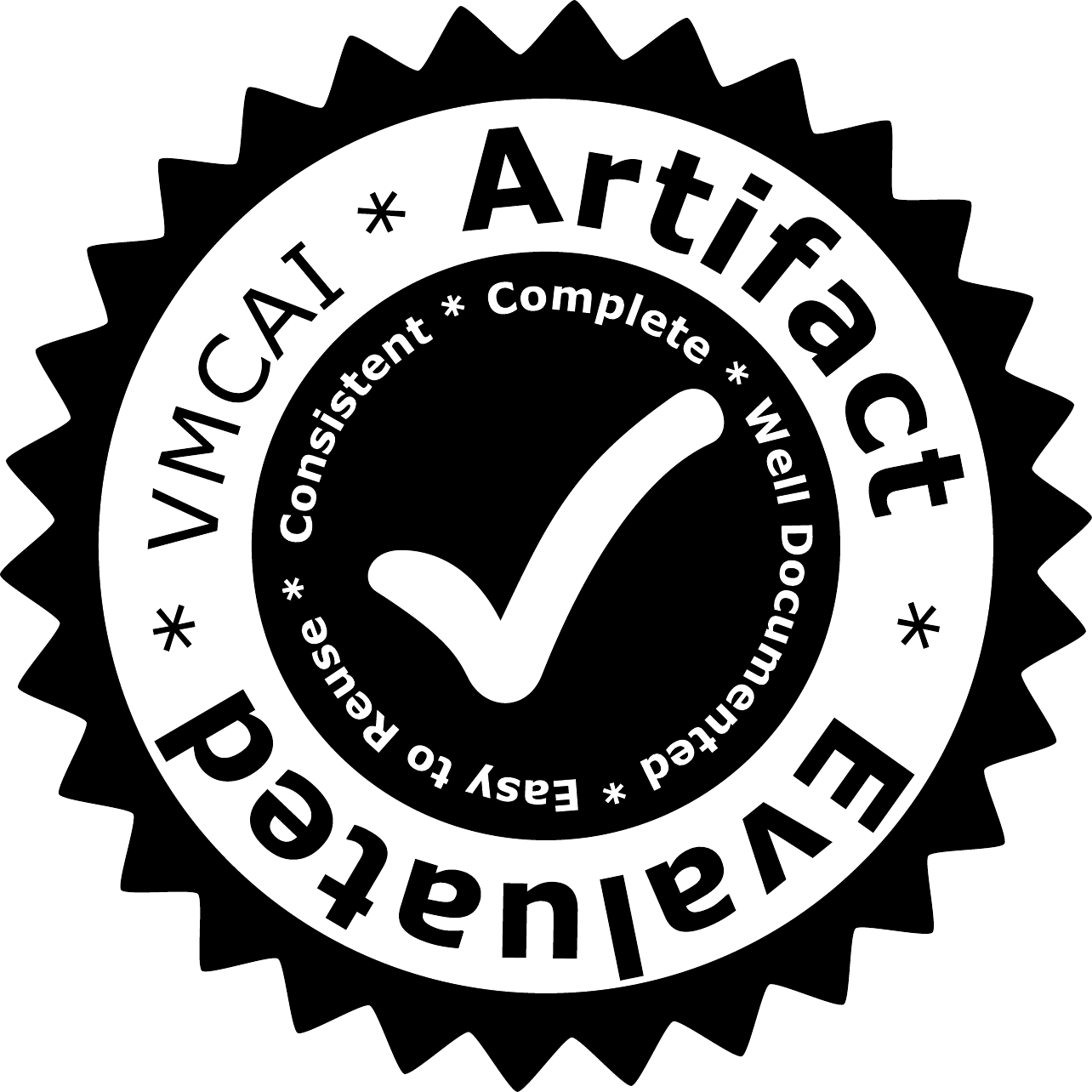}}}
\SetWatermarkAngle{0}

\usepackage{hyperref}

\usepackage{amsmath}
\usepackage{amsfonts}
\usepackage{amssymb} 
\usepackage{wasysym}
\usepackage{verbatimbox}
\usepackage{tikz}
\usepackage{stmaryrd}
\usepackage{wrapfig}
\usetikzlibrary{automata, positioning, arrows, shapes.geometric, calc}
\usepackage{multirow}
\usepackage{tabu} 
\usepackage{tabularx} 
\usepackage{booktabs}
\usepackage{ marvosym } 
\usepackage{csquotes}
\usepackage{adjustbox}
\usepackage{cleveref}
\usepackage{thm-restate}

\tikzstyle{myArrowStyle} = [shorten >=1pt,->,>=stealth',semithick]

\newcommand{\Z}{\mathbb{Z}}
\newcommand{\Dist}{\operatorname{Dist}}
\newcommand{\topartial}{\dashrightarrow}
\newcommand{\Set}[2]{\left\{\,#1\,\mid\,#2\,\right\}}
\newcommand{\set}[1]{\left\{\,#1\,\right\}}

\newcommand{\pcfp}{\mathfrak{P}}
\newcommand{\pcfpinit}{(\Loc,\, \Var,\, \dom,\, \Cmd,\, \iota)}
\newcommand{\pcfpinitprime}{(\Loc',\, \Var',\, \dom',\, \Cmd',\, \iota')}
\newcommand{\Var}{\mathsf{Var}}
\newcommand{\dom}{\mathsf{dom}}
\newcommand{\Loc}{\mathsf{Loc}}
\newcommand{\Cmd}{\mathsf{Cmd}}

\newcommand{\asgn}{\alpha}
\newcommand{\lhs}{\mathsf{lhs}}
\newcommand{\rhs}{\mathsf{rhs}}
\newcommand{\goalpred}{\vartheta}
\newcommand{\pcfptrans}[5]{#1 \xrightarrow{#2\,\rightarrow\,#3 \colon #4} #5}

\newcommand{\mdp}{\mathcal{M}}
\newcommand{\mdpinit}{(S,\, \Act,\, \sinit,\, P)}
\newcommand{\Act}{\mathsf{Act}}
\newcommand{\sinit}{\iota}
\newcommand{\mdptrans}[4]{#1 \xrightarrow{#2,~#3} #4}

\newcommand{\update}{u}
\newcommand{\seq}{\fatsemi}
\newcommand{\altupdate}{v}
\newcommand{\nop}{\mathsf{nop}}
\newcommand{\val}{\nu}
\newcommand{\guard}{\varphi}
\newcommand{\altguard}{\psi}

\newcommand{\Asg}{\mathsf{Asgn}}
\renewcommand{\lhs}{\mathsf{lhs}}
\renewcommand{\rhs}{\mathsf{rhs}}
\newcommand{\Unf}{\mathsf{Unf}}

\renewcommand{\wp}{\mathsf{wp}}

\newcommand{\prob}{\mathbb{P}}
\newcommand{\reach}{\lozenge}

\newcommand{\depon}{\rightarrow}

\newcommand{\prism}{\textsf{PRISM}}
\newcommand{\storm}{\textsf{storm}}
\newcommand{\jani}{\textsf{jani}}
\newcommand{\mcsta}{\textsf{mcsta}}

\newcommand{\benchmark}[1]{\textsc{#1}}

\newcommand{\hoare}[3]{\set{#1}\,#2\,\set{#3}}

\newcommand\xqed[1]{%
  \leavevmode\unskip\penalty9999 \hbox{}\nobreak\hfill
  \quad\hbox{#1}}
\newcommand\qedExample{\xqed{\small $\triangle$}}
\newcommand\qedProof{\xqed{\small$\square$}}

\makeatletter
\def\orcidID#1{\smash{\href{http://orcid.org/#1}{\protect\raisebox{-1.25pt}{\protect\includegraphics{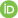}}}}}
\makeatother

\newtoggle{arxiv}
\toggletrue{arxiv}

\begin{document}
\title{Out of Control: Reducing Probabilistic Models by Control-State Elimination
    \thanks{This work is supported by the Research Training Group 2236 UnRAVeL, funded by the German Research Foundation.}
}
\titlerunning{Reducing Probabilistic Models by Control-State Elimination}
%
\author{Tobias Winkler\textsuperscript{(\Letter)}\orcidID{0000-0003-1084-6408} \and Johannes Lehmann\orcidID{0000-0001-7047-3813} \and \\
    Joost-Pieter Katoen\orcidID{0000-0002-6143-1926}}
\authorrunning{T. Winkler, J. Lehmann, and J.-P. Katoen}
%
\institute{RWTH Aachen University (Germany) \\
\email{tobias.winkler@cs.rwth-aachen.de} \\
\email{johannes.lehmann@rwth-aachen.de} \\ 
\email{katoen@cs.rwth-aachen.de}}
\maketitle            

\begin{abstract}

State-of-the-art probabilistic model checkers perform verification on explicit-state Markov models defined in a high-level programming formalism like the \prism\ modeling language.
Typically, the low-level models resulting from such program-like specifications exhibit lots of structure such as repeating subpatterns.
Established techniques like probabilistic bisimulation minimization are able to exploit these structures; however, they operate directly on the explicit-state model.
On the other hand, methods for reducing structured state spaces by reasoning about the high-level program have not been investigated that much.
In this paper, we present a new, simple, and fully automatic program-level technique to reduce the underlying Markov model.
Our approach aims at computing the summary behavior of adjacent locations in the program's control-flow graph, thereby obtaining a program with fewer ``control states''.
This reduction is immediately reflected in the program's operational semantics, enabling more efficient model checking.
A key insight is that in principle, each (combination of) program variable(s) with finite domain can play the role of the program counter that defines the flow structure.
Unlike most other reduction techniques, our approach is property-directed and naturally supports unspecified model parameters.
Experiments demonstrate that our simple method yields state-space reductions of up to 80\% on practically relevant benchmarks.

\end{abstract}

\section{Introduction}

\paragraph*{Modelling Markov models.}
Probabilistic model checking is a fully automated technique to rigorously prove correctness of a system model with randomness against a formal specification.
Its key algorithmic component is computing reachability probabilities on stochastic processes such as (discrete- or continuous-time) Markov chains and Markov Decision Processes.
These stochastic processes are typically described in some high-level modelling language.
State-of-the-art tools like \prism~\cite{prism}, \storm~\cite{storm} and \mcsta~\cite{mcsta} support input models specified in e.g., the \prism\ modeling language\footnote{\url{https://www.prismmodelchecker.org/manual/ThePRISMLanguage}}, PPDDL~\cite{ppddl}, a probabilistic extension of the planning domain definition language~\cite{DBLP:journals/jair/FoxL03}, the process algebraic language MoDeST~\cite{modest}, the \jani\ model exchange format~\cite{jani}, or the probabilistic guarded command language pGCL~\cite{DBLP:series/mcs/McIverM05}.
The recent tool from \cite{DBLP:conf/spin/FatmiCDWTB21} even supports verification of probabilistic models written in Java.

\paragraph{Model construction.}
Prior to computing reachability probabilities, existing model checkers explore all the program's reachable variable valuations and encode them into the state space of the operational Markov model.
Termination is guaranteed as variables are restricted to finite domains.
This paper proposes a simple reduction technique for this model construction phase that avoids unfolding the full model \emph{prior to} the actual analysis, thereby mitigating the state explosion problem. 
The basic idea is to unfold variables one-by-one---rather than all at once as in the standard pipeline---and apply analysis steps after each unfolding.
We detail this \emph{control-state reduction} technique for {probabilistic control-flow graphs} and illustrate its application to the \prism\ modelling language.
Its principle is however quite generic and is applicable to the aforementioned modelling formalisms.
Our technique is thus to be seen as a model simplification front-end for general purpose probabilistic model checkers.

\paragraph*{Approach.}
Technically our approach works as follows.
The principle is to unfold a (set of) variable(s) into the control state space, a technique inspired by static program analyses such as abstract interpretation~\cite{jeannet2003dynamic}.
The selection of which variables to unfold is property-driven, i.e., depending on the reachability or reward property to be checked.
We define the unfolding on probabilistic control-flow programs~\cite{dubslaff_pcfp} (PCFPs, for short) and simplify them using a technique that generalizes \emph{state elimination} in (parametric) Markov chains~\cite{daws}.
Our elimination technique heavily relies on classical \emph{weakest precondition reasoning}~\cite{DBLP:books/ph/Dijkstra76}.
This enables the elimination of several states at once from the underlying ``low-level'' Markov model while preserving \emph{exact} reachability probabilities or expected rewards.
\Cref{fig:nandplots} provides a visual intuition on the resulting model compression.

The choice of the variables and locations for unfolding and elimination, resp., is driven by heuristics.
In a nutshell, our unfolding heuristics prefers the variables that lead to a high number of control-flow locations without self-loops.
These loop-free locations are then removed by the elimination heuristics which gives preference to locations whose removal does not blow up the transition matrix of the underlying model. 
Unfolding and elimination steps are performed in an alternating fashion, but only until the PCFP size reaches a certain threshold.
After this, the reduction phase is complete and the transformed PCFP can be fed into a standard probabilistic model checker.

\begin{figure}[t]
    \centering
    \includegraphics[scale=0.3]{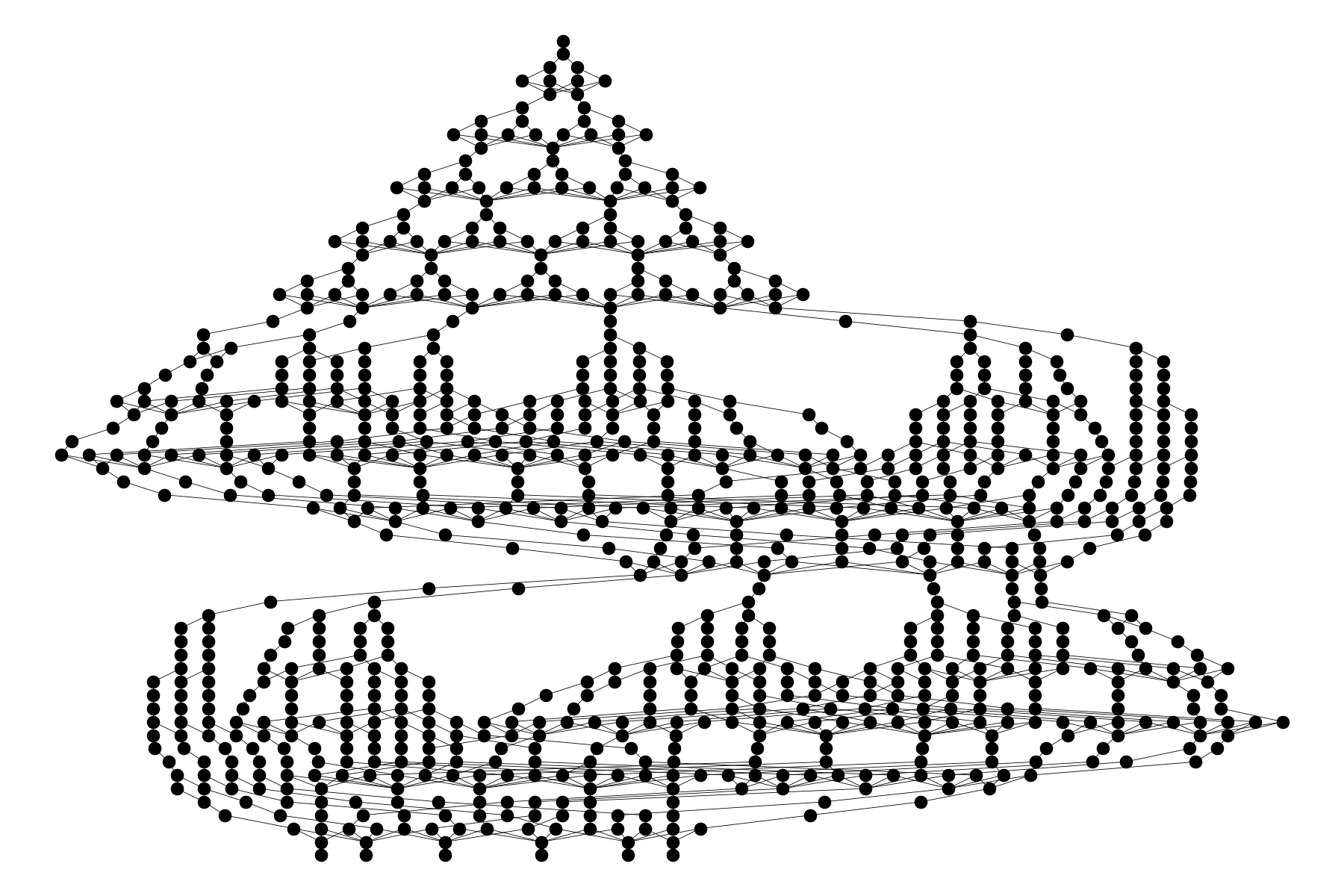}
    \hspace{5mm}
    \includegraphics[scale=0.3]{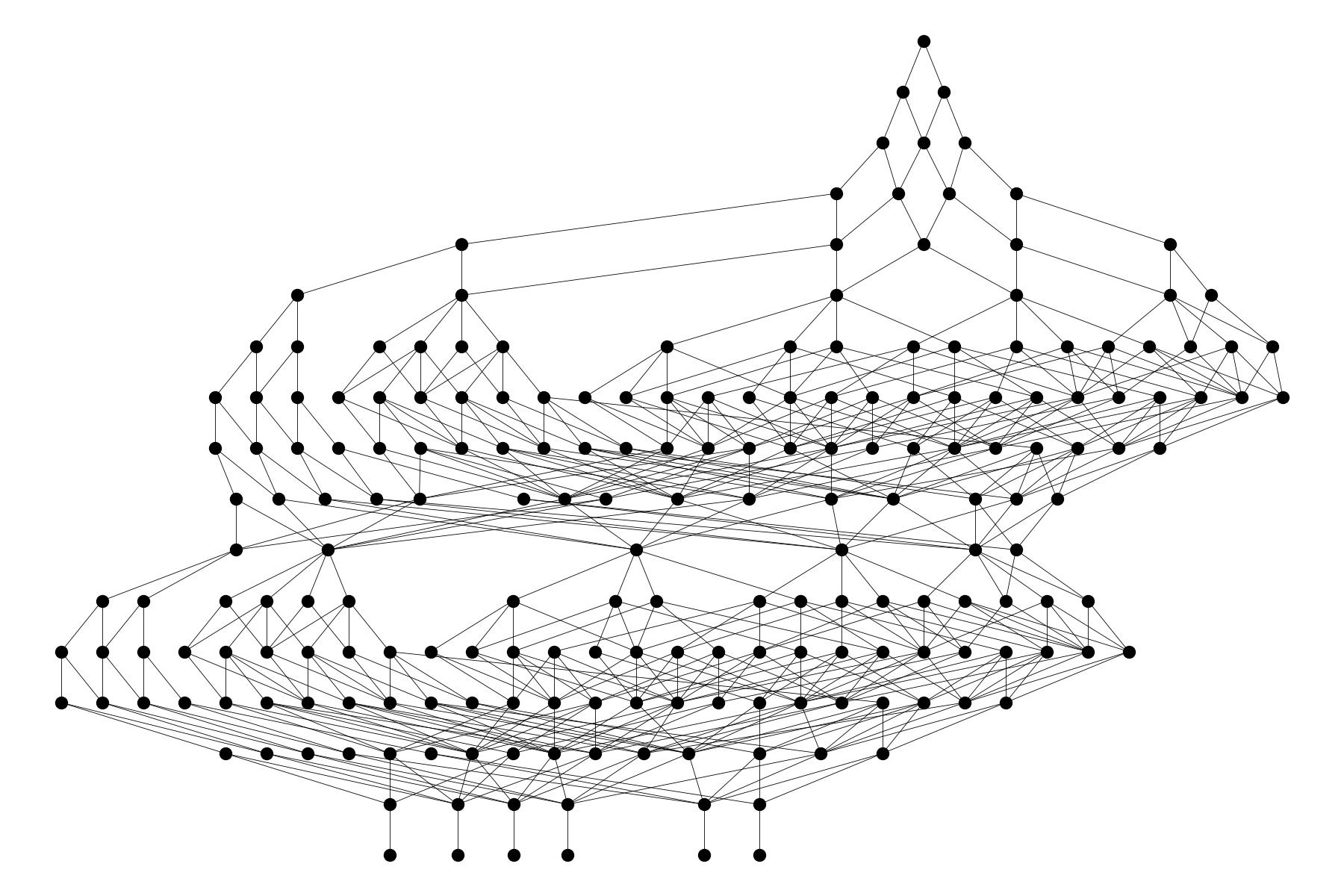}
    \caption{
        Left: Visualization of the original \benchmark{nand} model from~\cite{nand} (930 states, parameters 5/1).
        Transitions go from top to bottom.
        Right: The same model after our reduction ($207$ states).
        A single ``program counter variable'' taking at most 5 different values was unfolded and a total of three locations were eliminated thereafter.
        Note that the overall structure is preserved but several \emph{local} substructures such as the pyramidal shape at the top are compressed significantly.
        This behavior is typical for our approach.
    }
    \label{fig:nandplots}
\end{figure}

\paragraph*{Contributions.}
In summary, the main contributions of this paper are:
\begin{itemize}
\item A simple, widely applicable reduction technique that considers each program variable with finite domain as a ``program counter'' and selects suitable variables for unfolding into the control state space one-by-one.
\item A sound rule to eliminate control-flow locations in PCFPs in order to shrink the state space of the underlying Markov model while preserving \emph{exact} reachability probabilities or expected rewards.
\item Elimination in PCFPs---in contrast to Markov chains---is shown to have an exponential worst-case complexity.
\item An implementation in the probabilistic model checker \storm\ demonstrating the potential to significantly compress practically relevant benchmarks.
\end{itemize}

\paragraph*{Related work.}
The state explosion problem has been given top priority in both classical and probabilistic model checking.
Techniques similar to ours have been known for quite some time in the non-probabilistic setting~\cite{DBLP:conf/forte/DongR99,DBLP:conf/cav/KurshanLY02}.
Regarding probabilistic model checking, reduction methods on the state-space level include symbolic model checking using MTBDDs~\cite{symbolicPMC}, SMT/SAT techniques~\cite{DBLP:conf/vmcai/WimmerBB09,DBLP:conf/cav/BatzJKKMS20}, bisimulation minimization~\cite{DBLP:conf/tacas/KatoenKZJ07,DBLP:conf/tacas/ValmariF10,DBLP:conf/concur/JansenGTY20}, Kronecker representations~\cite{DBLP:journals/jlp/BuchholzKKT03,symbolicPMC} and 
partial order reduction~\cite{por_baier,por_dargenio}.
Language-based reductions include symmetry reduction~\cite{symmred}, bisimulation reduction using SMT on \prism\ modules~\cite{DBLP:conf/vmcai/DehnertKP13}, as well as abstraction-refinement techniques~\cite{DBLP:journals/fmsd/KattenbeltKNP10,DBLP:conf/tacas/HahnHWZ10,DBLP:conf/vmcai/WachterZ10}.
Our reductions on PCFPs are inspired by state elimination~\cite{daws}.
Similar kinds of reductions on probabilistic workflow nets have been considered in~\cite{DBLP:journals/pe/EsparzaHS17}.
Despite all these efforts, it is somewhat surprising that simple probabilistic control-flow reductions as proposed in this paper have not been investigated that much.
A notable exception is the recent work by Dubslaff \emph{et al.} that applies existing static analyses to control-flow-rich PCFPs~\cite{dubslaff_pcfp}. 
In contrast to our method, their technique yields bisimilar models and exploits a different kind of structure.

\paragraph{Organization of the paper.} 
\Cref{sec:example} starts off by illustrating the central aspects of our approach by example. 
\Cref{sec:prelims} defines PCFPs and their semantics in terms of MDPs. 
\Cref{sec:main} formalizes the reductions, proves their correctness and analyzes the complexity.
Our implementation in \storm\ is discussed in \Cref{sec:implementation}.
We present our experimental evaluation in \Cref{sec:experiments} and conclude in \Cref{sec:conclusion}.
\iftoggle{arxiv}{}{
A \emph{full version} of this paper including detailed proofs is available online~\cite{arxiv}.
}
\section{A Bird's Eye View}
\label{sec:example}
This section introduces a running example to illustrate our approach.
Consider a game of chance where a gambler starts with an initial budget of $\mathtt{x = N/2}$ tokens.
The game is played in rounds, each of which either increases or decreases the budget.
The game is lost once the budget has dropped to zero and won once it exceeds $\mathtt{N}$ tokens.
In each round, a fair coin is tossed:
If the outcome is tails, then the gambler loses one token and proceeds to the next round; on the other hand, if heads occurs, then the coin is flipped again.
If tails is observed in the second coin flip, then the gambler also loses one token; however, if the outcome is again heads then the gambler receives \emph{two} tokens.

In order to answer questions such as \enquote{Is this game fair?} (for a fixed $\mathtt{N}$), probabilistic model checking can be applied.
To this end, we model the game as the \prism\ program
in \Cref{fig:coingame_prism}.
We briefly explain its central components:
The first two lines of the \texttt{module} block are variable declarations.
Variable $\mathtt{x}$ is an integer with bounded domain and $\mathtt{f}$ is a Boolean.
The idea of $\mathtt{x}$ and $\mathtt{f}$ is to represent the current budget and whether the coin has to be flipped a second time, respectively.
The next three lines that each begin with $\texttt{[]}$ define \emph{commands} which are interpreted as follows:
If the \emph{guard} on the left-hand side of the arrow \texttt{->} is satisfied, then one of the updates on the right side is executed with its corresponding probability.
For instance, in the first command, $\mathtt{x}$ is decremented by one (and $\mathtt{f}$ is left unchanged) with probability $1/2$.
Otherwise $\mathtt{f}$ is set to true.
The order in which the commands occur in the program text is irrelevant.
If there is more than one command enabled for a specific valuation of the variables, then one of them is chosen non-deterministically.
Our example is, however, \emph{deterministic} in this regard since the three guards are mutually exclusive.

\begin{figure}[t]
	\centering
	\begin{verbbox}
dtmc
const int N;
module coingame
     x : [0..N+1] init N/2; 
     f : bool init false;
     [] 0<x & x<N & !f  ->  1/2: (x'=x-1)              + 1/2: (f'=true);
     [] 0<x & x<N & f   ->  1/2: (x'=x-1) & (f'=false) + 1/2: (x'=x+2) & (f'=false);
     [] x=0 | x>=N      ->  1:   (f'=false);
endmodule
	\end{verbbox}
    \begin{adjustbox}{max width=\textwidth-10pt - 1.33pt ,frame=0.66pt 5pt 0pt}
        \theverbbox 
    \end{adjustbox}
	\caption{The coin game as a \prism\ program. Variable $\mathtt{x}$ stands for the current budget.}
	\label{fig:coingame_prism}
\end{figure}

Probabilistic model checkers like \prism\ and \storm\ expand the above program as a Markov chain with approximately $2\mathtt{N}$ states.
This is depicted for $\mathtt{N} = 6$ at the top of \Cref{fig:example_unfolded}.
Given that we are only interested in the winning probability (i.e., to reach one of the two rightmost states), this Markov chain is equivalent to the smaller one on the bottom of \Cref{fig:example_unfolded}.
Indeed, \emph{eliminating} each dashed state in the lower row individually yields that the overall probability per round to go one step to the left is $3/4$ and $1/4$ to go two steps to the right.
On the program level, this simplification could have been achieved by summarizing the first two commands to
\begin{verbbox}
	[] 0<x & x<N  ->  3/4: (x'=x-1) + 1/4: (x'=x+2);
\end{verbbox}
\begin{center}
	\theverbbox
\end{center}
so that variable $\mathtt{f}$ is effectively removed from the program.

\begin{figure}[t]
	\centering 
	\begin{tikzpicture}[myArrowStyle, every state/.style={scale=0.5}, node distance= 6mm and 10mm]
	\node[state] (0!f) {};
	\node[state, right=of 0!f] (1!f) {};
	\node[state, right=of 1!f] (2!f) {};
	\node[state, right=of 2!f] (3!f) {};
	\node[state, right=of 3!f] (4!f) {};
	\node[state, right=of 4!f] (5!f) {};
	\node[state, accepting, right=of 5!f] (6!f) {};
	\node[state, accepting, right=of 6!f] (7!f) {};
	
	\node[state,white,below=of 0!f] (0f) {}; 
	\node[state, below=of 1!f,densely dashed] (1f) {};
	\node[state, right=of 1f,densely dashed] (2f) {};
	\node[state, right=of 2f,densely dashed] (3f) {};
	\node[state, right=of 3f,densely dashed] (4f) {};
	\node[state, right=of 4f,densely dashed] (5f) {};
	
	\node[above = 1mm of 0!f] {0};
	\node[above = 1mm of 1!f] {1};
	\node[above = 1mm of 2!f] {2};
	\node[above = 1mm of 3!f] {3};
	\node[above = 1mm of 4!f] {4};
	\node[above = 1mm of 5!f] {5};
	\node[above = 1mm of 6!f] {6};
	\node[above = 1mm of 7!f] {7};
	
	\node[left = 2mm of 0!f]  {\texttt{!f}};
	\node[left = 2mm of 0f]  {\texttt{f}};

	\draw[->] (1!f) -- (0!f);
	\draw[->] (2!f) -- (1!f);
	\draw[->] (3!f) -- (2!f);
	\draw[->] (4!f) -- (3!f);
	\draw[->] (5!f) -- (4!f);
	
	\draw[->] (1!f) -- (1f);
	\draw[->] (2!f) -- (2f);
	\draw[->] (3!f) -- (3f);
	\draw[->] (4!f) -- (4f);
	\draw[->] (5!f) -- (5f);
	
	\draw[->] (1f) -- (3!f);
	\draw[->] (2f) -- (4!f);
	\draw[->] (3f) -- (5!f);
	\draw[->] (4f) -- (6!f);
	\draw[->] (5f) -- (7!f);
	
	\draw[->] (1f) -- (0!f);
	\draw[->] (2f) -- (1!f);
	\draw[->] (3f) -- (2!f);
	\draw[->] (4f) -- (3!f);
	\draw[->] (5f) -- (4!f);
	
	\draw[->] (0!f) edge[loop below] (0!f);
	\draw[->] (6!f) edge[loop right] (6!f);
	\draw[->] (7!f) edge[loop below] (7!f);
	\end{tikzpicture}
	\vspace{4mm}
	
	\begin{tikzpicture}[myArrowStyle, every state/.style={scale=0.5}, node distance = 6mm and 10mm]
	\node[state] (0!f) {};
	\node[state, right=of 0!f] (1!f) {};
	\node[state, right=of 1!f] (2!f) {};
	\node[state, right=of 2!f] (3!f) {};
	\node[state, right=of 3!f] (4!f) {};
	\node[state, right=of 4!f] (5!f) {};
	\node[state, accepting, right=of 5!f] (6!f) {};
	\node[state, accepting,right=of 6!f] (7!f) {};

	\node[left = 2mm of 0!f]  {\phantom{!}\texttt{f}};

	\draw[->] (1!f) -- node[above] {\scriptsize $3/4$} (0!f);
	\draw[->] (2!f) -- node[above] {\scriptsize $3/4$}(1!f);
	\draw[->] (3!f) -- node[above] {\scriptsize $3/4$}(2!f);
	\draw[->] (4!f) -- node[above] {\scriptsize $3/4$}(3!f);
	\draw[->] (5!f) -- node[above] {\scriptsize $3/4$}(4!f);

	\draw[->] (1!f) edge[bend right] node[below] {\scriptsize $1/4$} (3!f);
	\draw[->] (2!f) edge[bend right] node[below] {\scriptsize $1/4$}(4!f);
	\draw[->] (3!f) edge[bend right] node[below] {\scriptsize $1/4$}(5!f);
	\draw[->] (4!f) edge[bend right] node[below] {\scriptsize $1/4$}(6!f);
	\draw[->] (5!f) edge[bend right] node[below] {\scriptsize $1/4$} (7!f);

	\draw[->] (0!f) edge[loop below] (0!f);
	\draw[->] (6!f) edge[loop right] (6!f);
	\draw[->] (7!f) edge[loop below] (7!f);
	\end{tikzpicture}
	\caption{Top: The Markov chain of the original coin game for $\mathtt{N}=6$. All transition probabilities (except on 	the self-loops) are $1/2$. Bottom: The Markov chain of the simplified model.}
	\label{fig:example_unfolded}
\end{figure}
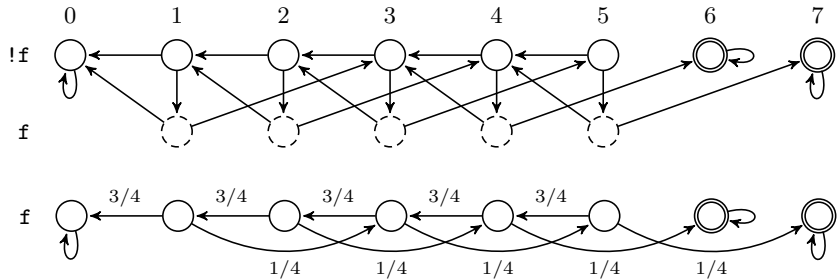

Obtaining such simplifications in an \emph{automated} manner is the main purpose of this paper. 
In summary, our proposed solution works as follows:
\begin{enumerate}
	\item
    First, we view the input program as a probabilistic control flow program (PCFP), which can be seen as a generalization of \prism\ programs from a single to multiple control-flow locations (\Cref{fig:coingame}, left).
    A \prism\ program (with a single module) is a PCFP with a unique control location.
    Imperative programs such as \emph{pGCL} programs~\cite{DBLP:series/mcs/McIverM05} can be regarded as PCFPs with roughly one location per line of code.
    %
	\item
    We then \emph{unfold} one or several variables into the location space, thereby interpreting them as \enquote{program counters}.
    We will discuss in \Cref{sec:unfolding} that---in principle---every variable can be unfolded in this way.
    The distinction between program counters and \enquote{data variables} is thus an informal one.
    This insight renders the approach quite flexible.
    In the example, we unfold $\mathtt{f}$ (\Cref{fig:coingame}, middle), but we stress that it is also possible to unfold $\mathtt{x}$ instead (for any fixed $\mathtt{N}$), even though this is not as useful in this case.
	\item
    The last and most important step is \emph{elimination}.
    Once sufficiently unfolded, we identify locations in the PCFP that can be eliminated.
    Our elimination rules are inspired by state elimination in Markov chains~\cite{daws}.
    In the example, we eliminate the location labeled $\mathtt{f}$.
    To this end, we try to eliminate all ingoing transitions of location $\mathtt{f}$.
    Applying the rules described in detail in \Cref{sec:main}, we obtain the PCFP shown in \Cref{fig:coingame} (right).
    This PCFP generates the reduced Markov chain in \Cref{fig:example_unfolded} (bottom).
	Here, location elimination has also reduced the size of the PCFP, but this is not always the case.
    In general, elimination adds more commands to the program while reducing the size of the generated Markov chain or MDP (cf.\ \Cref{sec:experiments}).
\end{enumerate}
These unfolding and elimination steps may be performed in an alternating fashion following the principle \emph{\enquote{unfold a bit, eliminate reasonably}}. Here, \enquote{reasonably} means that in particular, we must be careful to not blow up the underlying transition matrix (cf.\ \Cref{sec:implementation}).

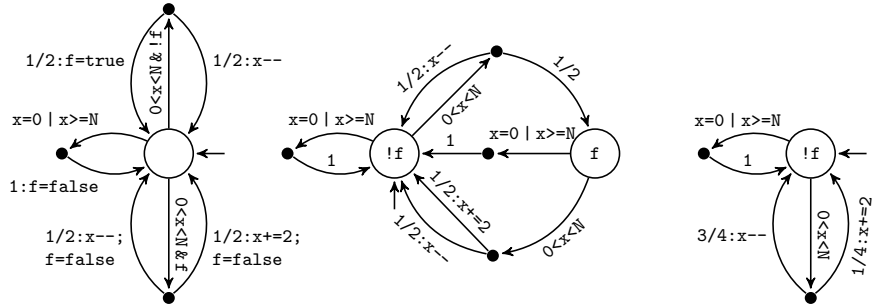
\begin{figure}[t]
	\centering
	\begin{tikzpicture}[myArrowStyle, node distance = 15mm, every node/.style={scale=0.8, font=\ttfamily},initial text=]
	\node[state,initial, initial where=right] (l) {};
	\node[circle,fill=black,inner sep=2pt, left = 10mm of l] (xeq0) {};
	\node[circle,fill=black,inner sep=2pt, above = of l] (nflip) {};
	\node[circle,fill=black,inner sep=2pt, below = of l] (flip) {};
	
	\draw[->] (l) edge[bend right] node[above left] {x=0\,|\,x>=N} (xeq0);
	\draw[->] (xeq0) edge[bend right] node[below left] {1:f=false} (l);
	
	\draw[->] (l) edge node[sloped, above] {0<x<N\,\&\,!f} (nflip);
	\draw[->] (nflip) edge[bend left=45] node[right,yshift=2mm] {1/2:x--} (l);
	\draw[->] (nflip) edge[bend right=45] node[left,yshift=2mm] {1/2:f=true} (l);
	
	\draw[->] (l) edge node[sloped, above] {0<x<N\,\&\,f} (flip);
	\draw[->] (flip) edge[bend left=45] node[left, align=left,yshift=-2mm] {1/2:x--;\\f=false} (l);
	\draw[->] (flip) edge[bend right=45] node[right, align=left,yshift=-2mm] {1/2:x+=2;\\f=false} (l);
	
	
	\node[state, right = 50mm of l] (f) {f};
	\node[state, left = 20mm of f,initial, initial where=below] (nf) {!f};
	
	\node[circle,fill=black,inner sep=2pt, left = 10mm of f] (xeq0f) {};
	\node[circle,fill=black,inner sep=2pt, left = 10mm of nf] (xeq0nf) {};
	
	\node[circle,fill=black,inner sep=2pt, above right = 15mm of nf] (nflip) {};
	\node[circle,fill=black,inner sep=2pt, below left = 15mm of f] (flip) {};
	
	\draw[->] (nf) edge[bend right] node[above] {x=0\,|\,x>=N} (xeq0nf);
	\draw[->] (xeq0nf) edge[bend right] node[above] {1} (nf);
	
	\draw[->] (f) edge node[above,yshift=1mm] {x=0\,|\,x>=N} (xeq0f);
	\draw[->] (xeq0f) edge node[above] {1} (nf);
	
	\draw[->] (nf) edge node[sloped, below] {0<x<N } (nflip);
	\draw[->] (nflip) edge[bend right] node[sloped, above] {1/2:x--} (nf);
	\draw[->] (nflip) edge[bend left] node[sloped, above ]{1/2} (f);
	
	\draw[->] (f) edge[bend left] node[sloped, below] {0<x<N } (flip);
	\draw[->] (flip) edge[bend left ] node[sloped, below] {1/2:x--} (nf);
	\draw[->] (flip) edge[bend left=0] node[sloped, above] {1/2:x+=2} (nf);
	
	
	\node[state, right= 22mm of f,initial, initial where=right] (nf) {!f};
	
	\node[circle,fill=black,inner sep=2pt, left =10mm  of nf] (xeq0) {};
	\node[circle,fill=black,inner sep=2pt, below = of nf] (flip) {};
	
	\draw[->] (nf) edge[bend right] node[above,xshift=-5mm,near start] {x=0\,|\,x>=N} (xeq0);
	\draw[->] (xeq0) edge[bend right] node[above] {1} (nf);
	
	
	
	\draw[->] (nf) edge node[sloped, above] {0<x<N} (flip);
	\draw[->] (flip) edge[bend left=45] node[left, align=left] {3/4:x--} (nf);
	\draw[->] (flip) edge[bend right=45] node[below,sloped] {1/4:x+=2} (nf);
	
	\end{tikzpicture}
	
	
	\caption{Left: The coin game as a single-location PCFP $\pcfp_{\mathit{game}}$. Middle: The PCFP after unfolding variable $\mathtt{f}$. Right: The PCFP after eliminating the location labeled $\mathtt{f}$.}
	\label{fig:coingame}
\end{figure}

Despite its simplicity, we are not aware of any other automatic technique that achieves the same or similar reductions on the coin game model.
In particular, bisimulation minimization is not applicable:
The bisimulation quotient of the Markov chain in \Cref{fig:example_unfolded} (top) is already obtained by merging just the two rightmost goal states.

Arguably, the program transformations in the above example could have been done by hand.
However, automation is crucial for our technique because the transformation makes the program harder to understand and obfuscates the original model's mechanics due to the removed intermediate control states.
Indeed, \emph{simplification} only takes place from the model checker's perspective but not from the programmer's.
Moreover, our transformations are rather tedious and error-prone, and may not always be that obvious for more complicated programs.
To illustrate this, we mention the work \cite{nand} where a \prism\ model of the von Neumann NAND multiplexing system was presented.
Optimizations with regard to the resulting state space were applied manually already at modeling time
\footnote{
    See paragraph 7 in \cite[Sec.~III A.]{nand}.
}.
Despite these (successful) manual efforts, our fully automatic technique can further shrink the state-space of the same model by $\approx 80\%$ (cf.\ \Cref{sec:experiments}).

\section{Technical Background on PCFPs}
\label{sec:prelims}

In this section, we review the necessary definitions of Markov Decision Processes (MDPs), Probabilistic Control Flow Programs (PCFPs), and reachability properties.
The set of probability distributions on a finite set $S$ is denoted $\Dist(S) = \Set{p \colon S \to [0,1]}{\sum_{s \in S}p(s) =1}$.
The set of (total) functions $A \to B$ is denoted $B^A$.

\subsubsection{Basic Markov Models}
An \emph{MDP} is tuple $\mdp = \mdpinit$ where $S$ is a finite set of states, $\sinit \in S$ is an initial state, $\Act$ is a finite set of action labels and $P \colon S \times \Act \topartial \Dist(S)$ is a (partial) probabilistic transition function.
We say that action $a \in \Act$ is \emph{available} at state $s \in S$ if $P(s,a)$ is defined.
We use the notation $\mdptrans{s}{a}{p}{s'}$ to indicate that $P(s,a)(s') = p$.
In the following, we write $P(s,a,s')$ rather than $P(s,a)(s')$.

A \emph{Markov chain} is an  MDP with exactly one available action at every state.
We omit action labels when considering Markov chains, i.e., the transition function of a Markov chain has type $P \colon S \to \Dist(S)$.
Given a Markov chain $\mdp$ together with a goal set $G \subseteq S$, we define the set of paths reaching $G$ as $\mathsf{Paths}(G) = \Set{s_0\ldots s_n \in S^n}{n \geq 0, s_0 = \sinit, s_n \in G, \forall i<n \colon s_i \notin G}$.
The \emph{reachability probability} of $G$ is
$
	\prob_{\mdp}(\reach G) = \sum_{\pi \in \mathsf{Paths}(G)} \prod_{i=0}^{l(\pi) - 1}P(\pi_i, \pi_{i+1})
$
where $l(\pi)$ denotes the length of a path $\pi$ and $\pi_i$ is the $i$-th state along $\pi$.
$\prob(\reach G)$ is always a well-defined probability (see e.g.~\cite[Ch.\ 10]{bk08} for more details).

A (memoryless deterministic) \emph{scheduler} of an MDP is a mapping $\sigma \in \Act^S$ with the restriction that action $\sigma(s)$ is available at $s$.
Each  scheduler $\sigma$ induces a Markov chain $\mdp^{\sigma}$ by retaining only the action $\sigma(s)$ at every $s \in S$.
Scheduler $\sigma$ is called \emph{optimal} if $\sigma = \operatorname{argmax}_{\sigma'}\prob_{{\mdp}^{\sigma'}}(\reach G)$ (or $\operatorname{argmin}$, depending on the context). In finite MDPs as considered here, there always exists an optimal memoryless and deterministic scheduler, even if the above $\operatorname{argmax}$ is taken over more general schedulers that may additionally use memory and/or randomization~\cite{puterman}.

\subsubsection{PCFP Syntax and Semantics}

We first define (guarded) commands. 
Let $\Var = \{x_1,\ldots,x_n\}$ be a set of integer-valued variables.
An \emph{update} is a set of assignments
\[
	\update \quad=\quad \set{ x_1' ~=~ f_1(x_1,\ldots,x_n),\quad \ldots\,, \quad x_n' ~=~  f_n(x_1,\ldots,x_n) }
\]
that are executed \emph{simultaneously}.
We assume that the expressions $f_i$ always yield integers.
An update $\update$ transforms a \emph{variable valuation} $\val \in \Z^\Var$ into a valuation $\val' = u(\val)$.
For technical reasons, we also allow \emph{chaining} of updates, that is, if $\update_1$ and $\update_2$ are updates, then $\update_1\seq\update_2$ is the update that corresponds to executing the updates in sequence: first $\update_1$ and then $\update_2$.
A \emph{command} is an expression
\[
	\guard \quad\rightarrow\quad p_1 \colon \update_1 ~+~ \ldots ~+~ p_k \colon \update_k ~,
\]
where $\guard$ is a \emph{guard}, i.e., a Boolean expression over program variables, $u_i$ are updates, and $p_i$ are non-negative real numbers such that $\sum_{i=1}^k p_i = 1$, i.e., they describe a probability distribution over the updates.
We further define \emph{location-guided} commands which additionally depend on \emph{control-flow locations} $l$ and $l_1,\ldots,l_k$:
\[
\guard,\, l \quad\rightarrow\quad p_1 \colon \update_1 \colon l_1 ~+~ \ldots ~+~ p_k \colon \update_k \colon l_k~.
\]
The intuitive meaning of a location-guided command is as follows:
It is enabled if the system is at location $l$ \emph{and} the current variable valuation satisfies $\guard$.
Based on the probabilities $p_1,\ldots,p_k$, the system then randomly executes one of the updates $\update_i$ and transitions to the next location $l_i$.
We use the notation $\pcfptrans{l}{\guard}{p_i}{\update_i}{l_i}$ to refer to such a possible \emph{transition} between locations.
We call location-guided commands simply \emph{commands} in the rest of the paper. 

\emph{Probabilistic Control Flow Programs} (PCFPs) combine several commands into a probabilistic program and constitute the formal basis of our approach:

\begin{definition}[PCFP]
	A PCFP is a tuple $\pcfp = \pcfpinit$ where $\Loc$ is a non-empty set of (control-flow) \emph{locations},
    $\Var$ is a set of integer-valued variables,
    $\dom \in \mathcal{P}(\mathbb{Z})^\Var$ is a \emph{domain} for each variable,
    $\Cmd$ is a set of \emph{commands} as defined above,
    and $\iota = (l_\iota, \val_\iota)$ is the initial location/valuation pair.
\end{definition}
This definition and our notation for commands are similar to \cite{dubslaff_pcfp}.
We also allow Boolean variables as \emph{syntactic sugar} by identifying $\mathtt{false} \equiv 0$ and $\mathtt{true} \equiv 1$.
We generally assume that $\Loc$ and all variable domains are \emph{finite} sets.
For a variable valuation $\val \in \Z^\Var$, we write $\val \in \dom$ if $\val(x) \in \dom(x)$ for all $x \in \Var$.
In some occasions, we consider only \emph{partial} valuations $\val \in \Z^{\Var'}$, where $\Var' \subsetneq \Var$.
We use the notations $\guard[\val]$ and $\update[\val]$ to indicate that all variables occurring in the guard $\guard$ (the update $\update$, respectively) are replaced according to the given (partial) valuation $\val$.
For updates, we also remove assignments whose left-hand side variables become a constant.
Recall that the notation $\update(\val)$ has a different meaning; it denotes the result of executing the update $\update$ on valuation $\val$.

The straightforward operational semantics of a PCFP is defined in terms of a Markov Decision Process (MDP).

\begin{definition}[MDP Semantics]
\label{def:mdpsemantics}
	For a PCFP $\pcfp = \pcfpinit$, we define the \emph{semantic MDP} $\mdp_{\pcfp} = \mdpinit$ as follows:
	\[
		S ~=~ \Loc \,\times\, \{\val \in \dom\} \,\cup\, \{\bot \},
        \qquad
        \Act ~=~ \Set{a_\gamma}{\gamma \in \Cmd},
        \qquad
        \iota ~=~ \langle l_\iota, \val_\iota \rangle
	\]
	and the probabilistic transition relation $P$ is defined according to the rules
	\[
		\frac{l_1 \xrightarrow{\guard \rightarrow p : \update} l_2 \,\wedge\, \val \models \guard \,\wedge\, \update(\val) \in \dom}{\langle l_1, \val \rangle \xrightarrow{a_{\gamma}, p} \langle l_2, \update(\val) \rangle},
		\hspace{5mm}
		\frac{l_1 \xrightarrow{\guard \rightarrow p : \update} l_2 \,\wedge\, \val \models \guard \,\wedge\, \update(\val) \notin \dom}{\langle l_1, \val \rangle \xrightarrow{a_{\gamma}, p} \bot}
	\]
	where $a_{\gamma} \in \Act$ is an action label that uniquely identifies the command $\gamma$ containing  transition $l_1 \xrightarrow{\guard \rightarrow p : \update} l_2$.
\end{definition}
An element $\langle l, \val \rangle \in  \Loc \times \{\val \in \dom\}$ is called a \emph{configuration}.
A PCFP is \emph{deterministic} if the MDP $\mdp_{\pcfp}$ is a Markov chain.
Moreover, we say that a PCFP is \emph{well-formed} if the out-of-bounds state $\bot$ is not reachable from the initial state and if there is at least one action available at each state of $\mdp_{\pcfp}$.
From now on, we assume that PCFPs are always well-formed.

\begin{example}
	The semantic MDP---a Markov chain in this case---of the two PCFPs in \Cref{fig:coingame} (left and middle) is given in \Cref{fig:example_unfolded} (top), and the one of the PCFP in \Cref{fig:coingame} (right) is depicted in \Cref{fig:example_unfolded} (bottom).
    \qedExample
\end{example}

\subsubsection{Reachability in PCFPs}
It is natural to describe a set of good (or bad) PCFP configurations by means of a predicate $\goalpred$ over the program variables which defines a set of target states in the semantic MDP $\mdp_{\pcfp}$.
We slightly extend this to account for information available from previous unfolding steps.
To this end, we will sometimes consider a labeling function $L \colon \Loc \to \mathbb{Z}^{\Var'}$ that assigns to each location an additional variable valuation $\val'$ over $\Var'$, a set of variables \emph{disjoint} to the actual programs variables $\Var$.
The idea is that $\Var'$ contains the variables that have already been unfolded (see \Cref{sec:unfolding} below for the details).
A predicate $\goalpred$ over $\Var \uplus \Var'$ describes the following goal set in the MDP $\mdp_{\pcfp}$:
\[
    G_{\goalpred} \quad = \quad \{\, \langle l, \val\rangle \,\mid\, l \in \Loc,~ \val \in \dom,\, (\val, L(l)) \models \goalpred \,\}
\]
where $(\val, L(l))$ is the variable valuation over $\Var \uplus \Var'$ that results from combining $\val$ and $L(l)$.

\begin{definition}[Potential Goal]
	Let $\pcfpinit$ be a PCFP labeled with valuations $L: \Loc \to \mathbb{\Z}^{\Var'}$ and let $\goalpred$ be a predicate over $\Var \uplus \Var'$. A location $l \in \Loc$ is called a \emph{potential goal} w.r.t.\ 
   $\goalpred$ if $\goalpred[L(l)]$ is satisfiable in $\dom$.
\end{definition}

\begin{example}
    \label{ex:potgoal}
	Consider the PCFP in \Cref{fig:coingame} (middle) with $\mathtt{N} =6$.
    Note that here, $\Var = \{\mathtt{x}\}$ and $\Var' = \{\mathtt{f}\}$.
    Let $\goalpred = (\texttt{x}\geq 6  \land \texttt{f}=\mathtt{false})$.
    Assume the labeling function $L(\mathtt{!f}) = \{f \mapsto \mathtt{false}\}$ and $L(\mathtt{f}) = \{f \mapsto \mathtt{true}\}$.
    Then the location labeled $\mathtt{!f}$ is a potential goal w.r.t.\ $\goalpred$ because $\goalpred[\mathtt{f} \mapsto \mathtt{false}] \equiv \mathtt{x\geq 6}$ is satisfiable.
    The other location $\mathtt{f}$ is no potential goal.
    \qedExample
\end{example}

In \Cref{sec:main} below, we introduce PCFP transformation rules that preserve reachability probabilities.
This is formally defined as follows:

\begin{definition}[Reachability Equivalence]
    Let $\pcfp_1$ and $\pcfp_2$ be PCFPs over the same set of variables $\Var$.
    For $i \in \{1,2\}$, let $L_i \colon \Loc_i \to \mathbb{Z}^{\Var'}$ be labeling functions on $\pcfp_i$.
    Further, let $\goalpred$ be a predicate over $\Var \uplus \Var'$.
    Then $\pcfp_1$ and $\pcfp_2$ are \emph{$\goalpred$-reachability equivalent} if
    \[
        \underset{\sigma}{\operatorname{opt}}~ \prob_{\mdp^{\sigma}_{\pcfp_1}}(\reach G_\goalpred)
        \quad=\quad
        \underset{\sigma}{\operatorname{opt}}~ \prob_{\mdp^{\sigma}_{\pcfp_2}}(\reach G_\goalpred)
    \]
    for both $\operatorname{opt} \in \{\min, \max\}$ and where $\sigma$ ranges of the class of memoryless deterministic schedulers for the MDPs $\mdp^{\sigma}_{\pcfp_1}$ and $\mdp^{\sigma}_{\pcfp_2}$, respectively.
\end{definition}

\begin{example}
    For all $\mathtt{N} \geq 0$, the PCFPs in \Cref{fig:coingame} (middle) and \Cref{fig:coingame} (right) with labeling functions as in \Cref{ex:potgoal} are reachability equivalent w.r.t.\ to $\goalpred = (\texttt{x}\geq \mathtt{N}  \land \texttt{f}=\mathtt{false})$.
    This follows from our intuitive explanation in \Cref{sec:example}, or alternatively from the formal rules to be presented in the folowing \Cref{sec:main}.
    \qedExample
\end{example}

\section{PCFP Reduction}
\label{sec:main}

We now describe our two main ingredients in detail: variable \emph{unfolding} and location \emph{elimination}.
Throughout this section, $\pcfp = \pcfpinit$ denotes an arbitrary well-formed PCFP.

\subsection{Variable Unfolding}
\label{sec:unfolding}

%
Let $\Asg$ be the set of all assignments that occur anywhere in the updates of $\pcfp$. For an assignment $\asgn \in \Asg$, we write $\lhs(\asgn)$ for the variable on the left-hand side and $\rhs(\asgn)$ for the expression on the right-hand side.
Let $x, y \in \Var$ be arbitrary. Define the relation $x \depon y$ (``$x$ depends on $y$'') as
\[
	x \depon y \qquad \iff \qquad \exists \asgn \in \Asg \colon\quad x = \lhs(\asgn) \quad\wedge\quad \rhs(\alpha) \text{ contains } y~.
\]
This syntactic dependency relation only takes updates but no guards into account.
This is, however, sufficient for our purpose.
We say that $x$ is \emph{(directly) unfoldable} if $\forall y \colon\, x \depon y \implies x = y$, that is, $x$ depends at most on itself.
\begin{example}
	Variables $\texttt{x}$ and $\texttt{f}$ in the PCFP in \Cref{fig:coingame} (left) are unfoldable.
    \qedExample
\end{example}
The rationale of this definition is as follows:
If variable $x$ is to be unfolded into the location space, then we must make sure that any update assigning to $x$ yields an explicit numerical value and hence an unambiguous location.
Formally, unfolding is defined as follows:

\begin{definition}[Unfolding]
	\label{def:unfolding}
    Let $x \in \Var$ be unfoldable. The \emph{unfolding} $\Unf(\pcfp, x)$ of $\pcfp$ with respect to $x$ is the PCFP $(\Loc',\, \Var \setminus \{x\},\, \dom,\, \Cmd',\, \iota')$ where
	\[
		\Loc' ~=~ \Loc \,\times\, \dom(x),
        \qquad
        \iota' ~=~ (\, \langle \, l_\iota,\, \val_\iota(x) \,\rangle,\, \val_\iota' \,)
	\]
	where $\val_\iota'(x) = \val_\iota(x)$ for all $x \in \Var'$, and $\Cmd'$ is defined according to the rule
	\[
		\frac{l \xrightarrow{\guard \to p:\update} l' \text{ in } \pcfp  \quad\wedge\quad \val\colon \{x\} \to  \dom(x)}{\langle\, l,\, \val(x) \,\rangle \,\xrightarrow{\guard[\val] \,\to\, p:\update[\val]}\, \langle\, l',\, \update(\val)(x) \,\rangle} ~.
	\]
\end{definition}
Recall that $\update[\val]$ substitutes all $x$ in $\update$ for $\val(x)$ while $\update(\val)$ applies $\update$ to valuation $\val$.
Note that even though $\val$ only assigns a value to $x$ in the above rule, we nonetheless have that $\update(\val)(x)$ is a well-defined integer in $\dom(x)$.
This is ensured by the definition of unfoldable and because $\pcfp$ is well-formed.
Unfolding preserves the semantics of a PCFP (up to renaming of states and action labels):

\begin{lemma}
\label{lemma:unfoldsemantics}
	For every unfoldable $x \in \Var$, we have $\mdp_{\Unf(\pcfp, x)} = \mdp_{\pcfp}$.
\end{lemma}
\begin{example}
	The PCFP in \Cref{fig:coingame} (middle) is the unfolding $\Unf(\pcfp_{\mathit{game}}, \mathtt{f})$ of the PCFP $\pcfp_{\mathit{game}}$ in \Cref{fig:coingame} (left) with respect to variable $\mathtt{f}$.
    \qedExample
\end{example}
In general, it is possible that no single variable of a PCFP is unfoldable. We offer two alternatives for such cases:
\begin{itemize}
	\item There always exists a \emph{set} $U \subseteq \Var $ of variables that can be unfolded \emph{at once} ($U = \Var$ in the extreme case). \Cref{def:unfolding} can be readily adapted to this case. Preferably small sets of unfoldable variables can be found by considering the bottom SCCs of the directed graph $(\Var, \depon)$.
	\item In principle, each variable can be made unfoldable by introducing further commands. Consider for instance a command $\gamma$ with an update $x' = y$. We may introduce $|\dom(y)|$ new commands by strengthening $\gamma$'s guard with condition ``$y = z$'' for each $z \in \dom(y)$ and substituting all occurrences of $y$ for the constant $z$. This transformation is mostly of theoretical interest as it may create a large number of new commands.
\end{itemize}

\subsection{Elimination}

For the sake of illustration, we first recall state elimination in Markov chains.
Let $s$ be a state of the Markov chain.
The first step is to eliminate all self-loops of $s$ by rescaling the probabilities accordingly (\Cref{fig:mc_state_elim}, left).
Afterwards, all ingoing transitions are redirected to the successor states of $s$ by multiplying the probabilities along each possible path (\Cref{fig:mc_state_elim}, right).
The state $s$ is then not reachable anymore and can be removed.
This preserves reachability probabilities in the Markov chain provided that $s$ was neither an initial nor goal state.
Note that state elimination may increase the total number of transitions.
In essence, state elimination in Markov chains is an automata-theoretic interpretation of solving a linear equation system by Gaussian elimination~\cite{param-synth-arxiv}.

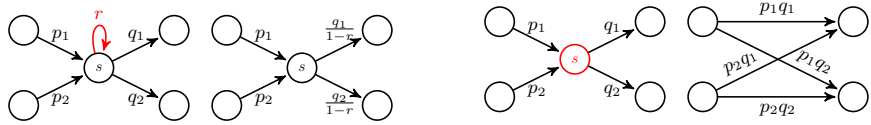
\begin{figure}[t]
	\centering
	\begin{tikzpicture}[myArrowStyle, node distance = 5mm and 10mm,on grid, every state/.style={scale=0.6}, every node/.style={scale=0.8}]
	\node[state] (s) {\large $s$};
	\node[state, above left=of s] (s1) {};
	\node[state, below left=of s] (s2) {};
	\node[state, above right=of s] (s3) {};
	\node[state, below right=of s] (s4) {};
	
	\draw[->] (s) edge[loop above, red] node[above] {$r$} (s);
	\draw[->] (s1) -- node[above] {$p_1$} (s);
	\draw[->] (s2) -- node[below] {$p_2$} (s);
	\draw[->] (s) -- node[above] {$q_1$} (s3);
	\draw[->] (s) -- node[below] {$q_2$} (s4);
	
	\node[state,right = 27mm of s] (s) {\large $s$};
	\node[state, above left=of s] (s1) {};
	\node[state, below left=of s] (s2) {};
	\node[state, above right=of s] (s3) {};
	\node[state, below right=of s] (s4) {};
	
	\draw[->] (s1) -- node[above] {$p_1$} (s);
	\draw[->] (s2) -- node[below] {$p_2$} (s);
	\draw[->] (s) -- node[above] {$\frac{q_1}{1-r}$} (s3);
	\draw[->] (s) -- node[below] {$\frac{q_2}{1-r}$} (s4);
	\end{tikzpicture}
	\hspace{10mm}
	\begin{tikzpicture}[myArrowStyle, node distance = 5mm and 10mm,on grid, every state/.style={scale=0.6}, every node/.style={scale=0.8}]
	\node[state,red] (s) {\large $s$};
	\node[state, above left=of s] (s1) {};
	\node[state, below left=of s] (s2) {};
	\node[state, above right=of s] (s3) {};
	\node[state, below right=of s] (s4) {};
	
	\draw[->] (s1) -- node[above] {$p_1$} (s);
	\draw[->] (s2) -- node[below] {$p_2$} (s);
	\draw[->] (s) -- node[above] {$q_1$} (s3);
	\draw[->] (s) -- node[below] {$q_2$} (s4);
	
	\node[right = 27mm of s] (s) {};
	\node[state, above left=of s] (s1) {};
	\node[state, below left=of s] (s2) {};
	\node[state, above right=of s] (s3) {};
	\node[state, below right=of s] (s4) {};
	
	\draw[->] (s1) -- node[above] {$p_1q_1$} (s3);
	\draw[->] (s2) -- node[above,sloped,near start] {$p_2q_1$} (s3);
	\draw[->] (s1) -- node[above,sloped, near end] {$p_1q_2$} (s4);
	\draw[->] (s2) -- node[below] {$p_2q_2$} (s4);
	\end{tikzpicture}
%
%
%
	\caption{State elimination in Markov chains. Left: Elimination of a self-loop. Right: Elimination of a state without self-loops. These rules preserve reachability probabilities provided that $s$ is neither initial nor a goal state.}
	\label{fig:mc_state_elim}
\end{figure}

In the rest of this section, we develop a \emph{location elimination rule for PCFPs} that generalizes state elimination in Markov chains.
Updates and guards are handled by weakest precondition reasoning which is briefly recalled below.
We then introduce a rule to remove single transitions, and show how it can be employed to eliminate \emph{self-loop-free}  locations.
For the (much) more difficult case of self-loop elimination, we refer to \iftoggle{arxiv}{\Cref{app:selfloops}}{the full version \cite{arxiv}} for the treatment of some special cases.
Handling general loops requires finding loop invariants which is notoriously difficult to automize.
Instead, the overall idea of this paper is to \emph{create self-loop-free locations by suitable unfolding}.

\subsubsection{Weakest Preconditions}
As mentioned above, our elimination rules rely on classical weakest preconditions which are defined as follows.
Fix a set $\Var$ of program variables with domains $\dom$.
Further, let $\update$ be an update and $\guard, \altguard$ be predicates over $\Var$.
We call $\hoare{\altguard}{\update}{\guard}$ a valid \emph{Hoare-triple} if
\[
\forall \val \in \dom \colon \quad \val \models \altguard \quad\implies\quad \update(\val) \models \guard~.
\]
The predicate $\wp(\update, \guard)$ is defined as the weakest $\altguard$ such that $\hoare{\altguard}{\update}{\guard}$ is a valid Hoare-triple and is called the \emph{weakest precondition} of $\update$ with respect to postcondition $\guard$.
Here, \enquote{weakest} is to be understood as \emph{maximal} in the semantic implication order on predicates.
Note that $\update(\val) \models \guard$ iff $\val \models \wp(\update, \guard)$.
It is well known~\cite{DBLP:books/ph/Dijkstra76} that for an update $\update = \set{ x_1' = f_1, ~\ldots,~ x_n' = f_n }$, the weakest precondition is given by
\[
\wp(\,\update,\, \guard\,) \quad=\quad \guard[\,x_1,\ldots,x_n \,\mapsto\, f_1,\ldots,f_n  \,] ~,
\]
i.e., all free occurrences of the variables $x_1,\ldots,x_n$ in $\guard$ are \emph{simultaneously} replaced by the expressions $f_1,\ldots,f_n$.
For example,
\[
\wp(\,\set{x' = y^2, y'= 5},\, x \geq y\,) \quad=\quad y^2 \geq 5 ~.  
\]
For chained updates $\update_1\seq\update_2$, we have $\wp(\update_1\seq\update_2, \guard) = \wp(\update_1, \wp(\update_2, \guard))$~\cite{DBLP:books/ph/Dijkstra76}.

\subsubsection{Transition Elimination}

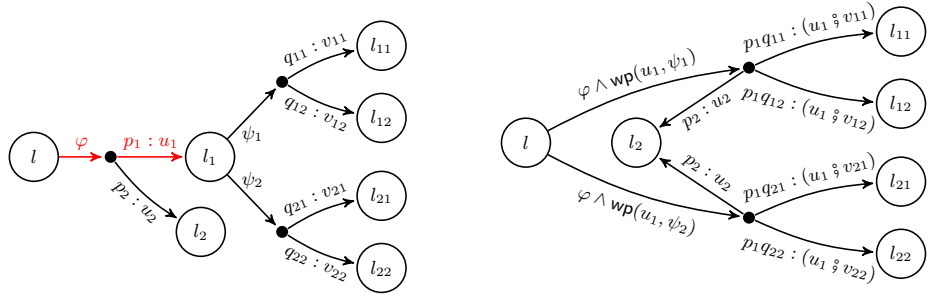
\begin{figure}[t]
	\centering
	\begin{tikzpicture}[myArrowStyle, node distance = 7mm and 9mm, every node/.style={scale=0.8}]
		\node[state] (lhat) {$l$};
		\node[circle, inner sep=2pt, right=6mm of lhat, fill=black] (phi1dot) {};
		\node[state, right=of phi1dot] (l1) {$l_1$};
		\node[state, below right=of phi1dot] (l2) {$l_2$};
		\node[circle, inner sep=2pt, above right=of l1, xshift=-3mm, fill=black] (psi1dot) {};
		\node[circle, inner sep=2pt, below right=of l1, xshift=-3mm, fill=black] (psi2dot) {};
		\node[state,  right=of psi1dot, yshift=6mm] (l11) {$l_{11}$};
		\node[state,  right=of psi1dot, yshift=-6mm] (l12) {$l_{12}$};
		\node[state,  right=of psi2dot, yshift=6mm] (l21) {$l_{21}$};
		\node[state,  right=of psi2dot, yshift=-6mm] (l22) {$l_{22}$};
		
		\draw[->,red] (lhat) edge node[above] {$\guard$} (phi1dot);
		
		\draw[->,red] (phi1dot) edge node[above] {$p_1:\update_1$} (l1);
		\draw[->] (phi1dot) edge[bend right=10] node[below, sloped] {$p_2:\update_2$} (l2);
		
		\draw[->] (l1) edge node[below,inner sep=2mm] {$\altguard_1$} (psi1dot);
		\draw[->] (l1) edge node[above,inner sep=2mm] {$\altguard_2$} (psi2dot);
		
		\draw[->] (psi1dot) edge[bend left=10] node[above,sloped] {$q_{11}:v_{11}$} (l11);
		\draw[->] (psi1dot) edge[bend right=10] node[below,sloped] {$q_{12}:v_{12}$} (l12);
		\draw[->] (psi2dot) edge[bend left=10] node[above,sloped] {$q_{21}:v_{21}$} (l21);
		\draw[->] (psi2dot) edge[bend right=10] node[below,sloped] {$q_{22}:v_{22}$} (l22);
	\end{tikzpicture}
	\hfill
	\begin{tikzpicture}[myArrowStyle, node distance = 7mm and 12mm, , every node/.style={scale=0.8}]
		\node[state] (lhat) {$l$};
		
		\node[state, right=8mm of lhat] (l2) {$l_2$};
	
		\node[circle, inner sep=2pt, above right=of l2, fill=black] (psi1dot) {};
		\node[circle, inner sep=2pt, below right=of l2, fill=black] (psi2dot) {};
		\node[state,  right=16mm of psi1dot, yshift=6mm] (l11) {$l_{11}$};
		\node[state,  right=16mm of psi1dot, yshift=-6mm] (l12) {$l_{12}$};
		\node[state,  right=16mm of psi2dot, yshift=6mm] (l21) {$l_{21}$};
		\node[state,  right=16mm of psi2dot, yshift=-6mm] (l22) {$l_{22}$};
		
		
		\draw[->] (lhat) edge[bend left=10] node[above,sloped] {$\guard \wedge \wp(\update_1, \altguard_1)$} (psi1dot);
		\draw[->] (lhat) edge[bend right=10] node[below,sloped] {$\guard \wedge \wp(\update_1, \altguard_2)$} (psi2dot);
		
		\draw[->] (psi1dot) edge node[below,sloped] {$p_2:\update_2$} (l2);
		\draw[->] (psi2dot) edge node[above,sloped] {$p_2:\update_2$} (l2);

		\draw[->] (psi1dot) edge[bend left=10] node[above,sloped] {$p_1 q_{11}:(u_1\seq v_{11})$} (l11);
		\draw[->] (psi1dot) edge[bend right=10] node[below,sloped] {$p_1 q_{12}:(u_1\seq v_{12})$} (l12);
		\draw[->] (psi2dot) edge[bend left=10] node[above,sloped] {$p_1 q_{21}:(u_1\seq v_{21})$} (l21);
		\draw[->] (psi2dot) edge[bend right=10] node[below,sloped] {$p_1 q_{22}:(u_1\seq v_{22})$} (l22);
	\end{tikzpicture}
	\caption{Transition elimination in PCFPs. Transition $l \xrightarrow{\guard \to p_1:\update_1} l_1$ is eliminated. The rule is correct even if the depicted locations are not pairwise distinct.}
	\label{fig:trans_elim}
\end{figure}

To simplify the presentation, we focus on the case of \emph{binary} PCFPs where locations have exactly two commands and commands have exactly two transitions (the general case is treated in \iftoggle{arxiv}{\Cref{app:general-elim}}{\cite{arxiv}}).
The following construction is depicted in \Cref{fig:trans_elim}.
Let $l \xrightarrow{\guard \to p_1:\update_1}l_1$ be the transition we want to eliminate and suppose that it is part of a command
\begin{align}
    \label{eq:cmdgamma}
	\gamma : \qquad l,~ \guard \quad\to\quad p_1 \colon \update_1 \colon l_1 ~+~ p_2 \colon \update_2 \colon l_2~.
\end{align}
Suppose that the PCFP is in a configuration $\langle l, \val \rangle$ where guard $\guard$ is enabled, i.e., $\val \models \guard$.
Intuitively, to remove the desired transition, we must jump with probability $p_1$ directly from $l$ to one of the possible destinations of $l_1$, i.e., either $l_{11}, l_{12}, l_{21}$ or $l_{22}$.
Moreover, we need to anticipate the---possibly non-deterministic---choice at $l_1$ already at $l$.
Note that guard $\altguard_1$ will be enabled at $l_1$ iff $\update_1(\val) \models \altguard_1$.
The latter is true iff $\val \models \wp(\update_1, \altguard_1)$.
Hence, if $\val \models \guard \land \wp(\update_1, \altguard_1)$, then we can choose to jump from $l$ directly to $l_{11}$ or $l_{12}$ with probability $p_1$.
The exact probabilities $p_1q_{11}$ and $p_1q_{12}$, respectively, are obtained by simply multiplying the probabilities along each path.
To preserve the semantics, we must also execute the updates found on these paths in the right order, i.e., either $\update_1\seq v_{11}$ or $\update_1\seq v_{12}$.
The situation is completely analogous for the other command with guard $\altguard_2$.

In summary, we apply the following transformation:
We remove the command $\gamma$ in \eqref{eq:cmdgamma} completely (and hence not only the transition $l \xrightarrow{\guard \to p_1:\update_1}l_1$) and replace it by \emph{two new commands} $\gamma_1$ and $\gamma_2$ which are defined as follows:
\begin{align*}
\gamma_i : \quad l,~ \guard \wedge \wp(\update_1, \altguard_i) ~\to~  p_2 \colon \update_2 \colon l_2 ~+~ \sum_{j=1}^2 p_1 q_{ij} \colon (\update_1 \seq \altupdate_{ij}) \colon l_{ij},\quad i \in \{1,2\}~.
\end{align*}
Note that in particular, this operation preserves deterministic PCFPs: If $\altguard_1$ and $\altguard_2$ are mutually exclusive, then so are $\wp(\update_1, \altguard_1)$ and $\wp(\update_1, \altguard_2)$. If the guards are not exclusive, then the construction transfers the non-deterministic choice from $l_1$ to $l$.

\renewcommand{\tt}{\texttt}
\begin{example}
	\label{ex:trans_elim}
	In the PCFP in \Cref{fig:coingame} (middle), we eliminate the transition
	\[\mathtt{
		!f \xrightarrow{0<x<N~\to~1/2 : nop} f ~.
	}\]
	The above transition is contained in the command
    \[
        \texttt{!f, 0\,<\,x\,<\,N} \quad\to\quad \tt{1/2} : \tt{nop} : \tt{f} ~+~ \tt{1/2} : \tt{x--} : \tt{!f} ~~.
    \]
	The following two commands are available at location $\mathtt{f}$:
	\begin{align*}
         \texttt{f, x=0\,|\,x\,>=\,N} & \quad\to\quad \tt{1} : \tt{nop} : \tt{!f} \\
		 \texttt{f, 0\,<\,x\,<\,N} & \quad\to\quad \tt{1/2} : \tt{x+=2} : \tt{!f} ~+~ \tt{1/2} : \tt{x--} : \tt{!f} ~~.
	\end{align*}
	Note that $\wp(\mathtt{nop}, \altguard) = \altguard$ for any guard $\altguard$. According to the construction in \Cref{fig:trans_elim}, we add the following two new commands to location $\mathtt{!f}$:
	\begin{align*}
	 \tt{!f, 0\,<\,x\,<\,N \& (\,x=0\,|\,x\,>=\,N\,)} & \quad\to\quad  \tt{1/2} : \tt{nop} : \tt{!f} ~+~ \tt{1/2} : \tt{x--} : \tt{!f} \\ 
	 \tt{!f, 0\,<\,x\,<\,N \& 0\,<\,x\,<\,N} & \quad\to\quad \tt{1/2} : \tt{x--} : \tt{!f} ~+~ \tt{1/4} : \tt{x--} : \tt{!f}\\
									& \quad\phantom{\to}\quad ~+~ \tt{1/4} : \tt{x=x+2} : \tt{!f}~~.
	\end{align*}
	The guard of the first command is unsatisfiable so that the whole command can be discarded.
    The second command can be further simplified to \[
		\tt{!f, 0\,<\,x\,<\,N} \quad\to\quad \tt{3/4} : \tt{x--} : \tt{!f} ~+~ \tt{1/4} :  \tt{x=x+2} : \tt{!f}~~.
	\]
	Removing unreachable locations yields the PCFP in \Cref{fig:coingame}~(right).
    \qedExample
\end{example}

Regarding the correctness of transition elimination, the intuitive idea is that the rule preserves reachability probabilities if location $l_1$ is \emph{not} a potential goal.
Recall that potential goals are locations for which we do not know whether they contain goal states when fully unfolded.
Formally, we have the following:

\begin{lemma}
	\label{lem:trans_elim_correct}
	Let $l_1 \in \Loc \setminus \{l_\iota\}$ be no potential goal with respect to goal predicate $\goalpred$ and let $\pcfp'$ be obtained from $\pcfp$ by eliminating transition $l \xrightarrow{\guard \to p_1:\update_1} l_1$ according to \Cref{fig:trans_elim}.
    Then $\pcfp$ and $\pcfp'$ are $\goalpred$-reachability equivalent.
\end{lemma}
\begin{proof}[Sketch]
    This follows by extending Markov chain transition elimination to MDPs and noticing that the semantic MDP $\mdp_{\pcfp'}$ is obtained from $\mdp_{\pcfp}$ by applying transition elimination repeatedly\iftoggle{arxiv}{ (see \Cref{proof:trans_elim_correct})}{, see \cite{arxiv} for the details}.
    \qedProof
\end{proof}

\subsubsection{Location Elimination}
We say that location $l \in  \Loc$ has a \emph{self-loop} if there exists a transition $l \xrightarrow{\guard\to p: \update} l$.
In analogy to state elimination in Markov chains, we can directly remove any location \emph{without self-loops} by applying the elimination rule to its ingoing transitions.
However, the case $l_1 = l_2$ in \Cref{fig:trans_elim} needs to be examined carefully as eliminating $l \xrightarrow{\guard\to p_1: \update_1} l_1$ actually \emph{creates two new} ingoing transitions to $l_1 = l_2$.
Termination of the algorithm is thus not immediately obvious.
Nonetheless, even for general (non-binary) PCFPs, the following holds:

\begin{theorem}[Correctness of Location Elimination]
	\label{thm:locelim}
	If $l \in \Loc \setminus \{l_\iota\}$ has no self-loops and is not a potential goal w.r.t.\ goal predicate $\goalpred$, then the algorithm
	\[
		\emph{\textbf{while }} (\exists~ l' \xrightarrow{\guard\to p: \update} l \text{ in }\pcfp) ~~ \{~ \text{eliminate } l' \xrightarrow{\guard\to p: \update} l ~\}
	\]
	terminates with a $\goalpred$-reachability equivalent PCFP $\pcfp'$ where $l$ is unreachable.
\end{theorem}
The following notion is helpful for proving termination of the above algorithm:
\begin{definition}[Transition Multiplicity]
    \label{def:mult}
    Given a transition $l' \xrightarrow{\guard\to p: \update} l$ contained in command $\gamma$, we define its \emph{multiplicity} $m$ as the total number of transitions in $\gamma$ that also have destination $l$.
\end{definition}
For instance, if $l_1 = l_2$ in \Cref{fig:trans_elim}, then transition $l \xrightarrow{\guard\to p_1: \update_1} l_1$ has multiplicity $m = 2$.
If $l_1 \neq l_2$, then it has multiplicity $m = 1$.
\begin{proof}[of \Cref{thm:locelim}]
	With \Cref{lem:trans_elim_correct} it only remains to show termination.
    We directly prove the general case where $\pcfp$ is non-binary.
    Suppose that $l$ has $k$ commands.
    Eliminating a transition entering $l$ with multiplicity $1$ does not create any new ingoing transitions (as $l$ has no self-loops).
    On the other hand, eliminating a transition with multiplicity $m > 1$ creates $k$ new commands, each with $m-1$ ingoing transitions to $l_1$.
    Thus, as the multiplicity strictly decreases, the algorithm terminates.
    \qedProof
\end{proof}

We now analyze the complexity of the algorithm in \Cref{thm:locelim} in detail.

\begin{restatable}[Complexity of Location Elimination]{theorem}{thmcomplexity}
    \label{thm:complexity}
   Let $l \in \Loc \setminus \{l_\iota\}$ be a location without self-loops.
   Let $k$ be the number of commands available at $l$.
   Further, let $n$ be the number of distinct commands in $\Cmd$ that have a transition with destination $l$,
   and suppose that each such transition has multiplicity at most $m$.
   Then the location elimination algorithm in \Cref{thm:locelim} applied to $l$ has the following properties:
   \begin{itemize}
       \item It terminates after at most $n(k^m {-}1)/(k{-}1)$ iterations.
       \item It creates at most $\mathcal{O}(nk^{m})$ new commands.
       \item There exist PCFPs where it creates at least $\Omega(n2^m)$ new \emph{distinct} commands with satisfiable guards.
   \end{itemize}
\end{restatable}
\begin{proof}[Sketch]
    We only consider the case $n=1$ here, the remaining details are treated in \iftoggle{arxiv}{\Cref{proof:complexity}}{\cite{arxiv}}.
    We show the three items independently:
    \begin{itemize}
        \item The number $I(m)$ of iterations of the algorithm in \Cref{thm:locelim} applied to location $l$ satisfies the recurrence $I(1)=1$ and $I(m)= 1 + kI(m-1)$ for all $m > 1$ since eliminating a transition with multiplicity $m > 1$ yields $k$ new commands with multiplicity $m-1$ each.
        The solution of this recurrence is $I(m) = \sum_{i=0}^{m-1} k^i = (k^m {-}1)/(k{-}1)$ as claimed.
        \item For the upper bound on the number of new commands, we consider the execution of the algorithm in the following stages:
        In stage 1, there is a single command with multiplicity $m$.
        In stage $j$ for $j > 1$, the commands from the previous stage are transformed into $k$ new commands with multiplicity $m-j+1$ each.
        In the final stage $m$, there are thus $k^{m-1}$ commands with multiplicity $1$ each.
        Eliminating all of them yields $k \cdot k^{m-1} = k^m$ new commands after which the algorithm terminates.
        \item Consider the PCFP $\pcfp$ in \Cref{fig:exponential} where $k=2$.
        Intuitively, location elimination must yield a PCFP $\pcfp'$ with $2^m$ commands available at location $l'$ because every possible combination of the updates $y_i' = 1$, $i=1,\ldots,m$, may result in enabling either of the two guards at $l$.
        Indeed, for each such combination, the guard which is enabled depends on the values of $x_1,\ldots,x_m$ at location $l'$.
        Thus in the semantic MDP $\mdp_{\pcfp'}$, for every variable valuation $\val$ with $\val(y_i) = 0$ for all $i =1,\ldots,m$, the probabilities $P(\langle l', \val \rangle, \langle l_1, \vec{0} \rangle)$ are \emph{pairwise distinct}.
        This implies that $\pcfp'$ must have $2^m$ commands (with satisfiable guards) at $l'$.
        \qedProof
        %
    \end{itemize}
    
\end{proof}

\newcommand{\figexponential}{
    \begin{tikzpicture}[myArrowStyle,every node/.style={scale=1.0}]
        \node[state] (l') {$l'$};
        \node[circle, fill=black, inner sep=2pt, right=of l'] (dot) {};
        \node[right=of dot,yshift=1mm] (dots) {$\vdots$};
        \node[state, right=20mm of dot] (l) {$l$};
        \node[circle, fill=black, inner sep=2pt, right= 20mm of l, yshift=5mm] (dot2) {};
        \node[state, right= 38mm of dot2] (l1) {$l_1$};
        \node[state, right= 15mm of l, yshift=-5mm] (l2) {$l_2$};
        
        \draw[->] (l') -- node[above] {$true$} (dot);
        \draw[->] (dot) edge[bend left] node[above] {$\frac{c}{2^1} \colon y_1' = 1$}(l);
        \draw[->] (dot) edge[bend right] node[below] {$\frac{c}{2^m} \colon y_m' = 1$}(l);
        \draw[->] (l) edge node[above, sloped] {$\bigvee_{i=1}^m (x_i \wedge y_i)$}(dot2);
        \draw[->] (dot2) edge node[above, sloped] {$\Set{x_i'=0,\, y_i'=0}{1{\leq} i {\leq} m}$}(l1);
        \draw[->] (l) edge node[below, sloped] {$\neg (...)$}(l2);
    \end{tikzpicture}   
}

\begin{figure}[t]
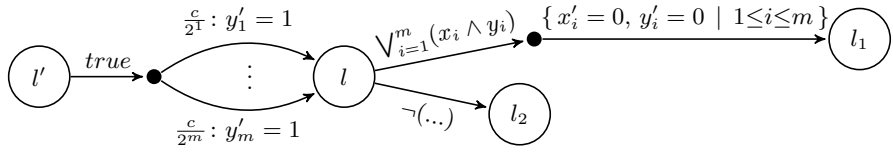

    \centering
    \figexponential
    \caption{
        The PCFP $\pcfp$ used for the lower bound in \Cref{thm:complexity}.
        The transitions from $l'$ to $l$ have multiplicity $m$ each.
        Variables $x,y$ have Boolean domain, $c$ is a normalizing constant.
    }
    \label{fig:exponential}
\end{figure}

\section{Implementation}
\label{sec:implementation}

\subsubsection{Overview}
We have implemented%
\iftoggle{arxiv}{%
    \footnote{%
        Code available at: \url{https://github.com/moves-rwth/storm/tree/master/src/storm/storage/jani/localeliminator}
    }%
}{}
our approach in the probabilistic model checker \storm~\cite{storm}.
Technically, instead of defining custom data structures for our PCFPs, we operate directly on models in the \jani\ model exchange format~\cite{jani}.
\storm\ accepts \jani\ models as input and also supports conversion from \prism\ to \jani.
The PCFPs described in this paper are a subset of the models expressible in \jani.
Other \jani\ models such as timed or hybrid automata are not in the scope of our implementation.
In practice, we use our algorithms as a \emph{simplification front-end}, i.e., we apply just a handful of unfolding and elimination steps and then fall back to \storm's default engine.
This is steered by heuristics that we explain in detail further below.

\subsubsection{Features}
Apart from the basic PCFPs treated in the previous sections, our implementation supports the following more advanced \jani\ features:
\begin{itemize}
    \item \emph{Parameters.}
    It is common practice to leave key quantities in a high-level model undefined and then analyze it for various instantiations of those parameters (as done in most of the \prism\ case studies\footnote{\url{https://www.prismmodelchecker.org/casestudies/}}); or synthesize in some sense suitable parameters~\cite{DBLP:conf/cav/DehnertJJCVBKA15,DBLP:conf/atva/QuatmannD0JK16,param-synth-arxiv}.
    Examples include undefined probabilities or undefined variable bounds like \texttt{N} in the \prism\ program in \Cref{fig:coingame_prism}.
    Our approach can naturally handle such parameters and is therefore particularly useful in situations where the model is to be analyzed for several parameter configurations.
    Virtually, the only restriction is that we cannot unfold variables with parametric bounds.
    \item \emph{Rewards.}
    Our framework can be easily extended to accommodate expected-reward-until-reachability properties (see e.g.~\cite[Def.~10.71]{bk08} for a formal definition).
    The latter are also highly common in the benchmarks used in the quantitative verification literature~\cite{qcomp}.
    Formally, in a \emph{reward PCFP}, each transition is additionally equipped with a non-negative reward that can either be a constant or given as an expression in the program variables.
    Technically, the treatment of rewards is straightforward:
    Each time we multiply the probabilities of two transitions in our transition elimination rule (\Cref{fig:trans_elim}), we \emph{add} their corresponding rewards.
    \item \emph{Parallel composition.}
    PCFPs can be extended by action labels to allow for synchronization of various parallel PCFPs.
    This is standard in model checking (e.g.~\cite[Sec.~2.2.2]{bk08}).
    We have implemented two approaches for dealing with this:
    (1) A ``flat'' product model is constructed first.
    This functionality is already shipped with the \storm\ checker.
    This approach is restricted to compositions of just a few modules as the size of the resulting product PCFP is in general exponential in the number of modules.
    Nonetheless, in many practical cases, flattening leads to satisfactory results (cf.\ \Cref{sec:experiments}).
    (2) Control-flow elimination is applied to each component individually. Here, we may only eliminate \emph{internal}, i.e. non-synchronizing commands, and we forbid shared variables.
    Otherwise, we would alter the resulting composition.
    \item \emph{Probability expressions.}
    Without changes, all of the theory presented so far can be applied to PCFPs with probability expressions like $|x|/(|x|+1)$ over the program variables instead of constant probabilities only.
    Expressions that do not yield correct probabilities are considered modeling errors.
\end{itemize}

\subsubsection{Heuristics}
The choice of the next variable to be unfolded and the next location to be eliminated is driven by heuristics.
The overall goal of the heuristics is to eliminate as many locations as possible while maintaining a reasonably sized PCFP.
This is controlled by two configurable parameters, $L$ and $T$.
The heuristics alternates between unfolding and elimination (see the diagram in \Cref{fig:heu}).

To find a suitable variable for unfolding, the heuristics first analyzes the dependency graph defined in \Cref{sec:unfolding}.
It then selects a variable based on the following static analysis:
For each unfoldable variable $x$, the heuristics considers each command $\gamma$ in the PCFP and determines the percentage $p(\gamma, x)$ of $\gamma$'s transitions that have an update with writing access to $x$.
Each variable is then assigned a \emph{score} which is defined as the average percentage $p(\gamma, x)$ over all commands of the PCFP.
The intuition behind this technique is that variables which are changed in many commands are more likely to create self-loop free locations when unfolded.
We consider the percentage for each command individually in order to not give too much weight to commands with many transitions.
Unfolding is only performed if the current PCFP has at most $L$ locations.
By default, $L = 10$ which in practice often leads to unfolding just two or three variables with small domains.

After unfolding a variable, the heuristics tries to eliminate self-loop-free locations that are no potential goals.
The next location to be eliminated is selected by estimating the number of new commands that would be created by the algorithm.
Here, we rely on the theoretical results from \Cref{thm:complexity}:
In particular, we take the \emph{multiplicity} (cf.\ \Cref{def:mult}) of ingoing transitions into account which may cause an exponential blowup.
We use the estimate $\mathcal{O}(nk^m)$ from \Cref{thm:complexity} as an approximation for the elimination complexity; determining the \emph{exact} complexity of each possible elimination is highly impractical.
We only eliminate locations whose estimated complexity is at most $T$, and we eliminate those with lowest complexity first.
By default, $T = 10^4$.

\begin{figure}[t]
    \centering
    \begin{tikzpicture}[myArrowStyle,initial where=above, every node/.style={align=center,scale=0.6,inner sep=7pt},node distance=8mm and 12mm]
        \node[initial, rectangle, rounded corners, draw=black] (depgraph) {Build\\dependency\\graph};
        \node[below=of depgraph] (done) {done};
        \node[right=of depgraph,ellipse,draw=black] (loctest) {$|\Loc| < L$};
        \node[below =of loctest, rectangle, rounded corners, draw=black] (unfold) {Unfold $x \in \Var$ \\ with max. \\ score};
        \node[right=20mm of loctest, xshift=3mm, rectangle, rounded corners, draw=black] (elim) {Eliminate $l \in \Loc$ \\ with min.\ compl.};
        \node[right=20mm of unfold, ellipse, draw=black,inner sep=0.5pt] (transtest) {$\exists\, l \in \Loc \setminus \{l_{\iota}\}$ s.t.\\ -- $l$ loop-free \\ -- $l$ no pot.\ goal \\ -- est.\ compl.\ $< T$};

        \draw (depgraph) -- (loctest);
        \draw (loctest) -- node[right] {yes} (unfold);
        \draw (loctest) -- node[left] {no} (done);
        \draw (unfold) edge[bend right=20] (transtest);
        \draw (transtest) edge[bend right=20] node[above] {no} (loctest);
        \draw (transtest) edge[bend right] node[right] {yes} (elim);
        \draw (elim) edge[bend right] (transtest);
    
        \draw[black,dotted,rounded corners=10pt] 
            ($(loctest.north west) + (-7mm,7mm)$) 
         -- ($(loctest.north east) + (7mm,7mm)$)
         -- ($(unfold.south east) + (4.5mm,-3mm)$)
         -- ($(unfold.south west) + (-4.5mm,-3mm)$)
         -- cycle;
        \node[above=0mm of loctest] {\textit{Unfolding}};

        \draw[black,dotted,rounded corners=10pt] 
            ($(elim.north west) + (-5mm,5mm)$) 
         -- ($(elim.north east) + (6mm,5mm)$)
         -- ($(transtest.south east) + (6.5mm,-3.5mm)$)
         -- ($(transtest.south west) + (-6.5mm,-3.5mm)$)
         -- cycle;
        \node[above=0mm of elim] {\textit{Elimination}};
    \end{tikzpicture}
    \caption{
        Our heuristics alternates between unfolding and elimination steps.
        The next unfold is determined by selecting a variable with maximal \emph{score} as computed by a static analysis (see main text).
        Loop-free non-potential goal locations are then eliminated until the next elimination has a too high estimated complexity.
    }
    \label{fig:heu}
\end{figure}
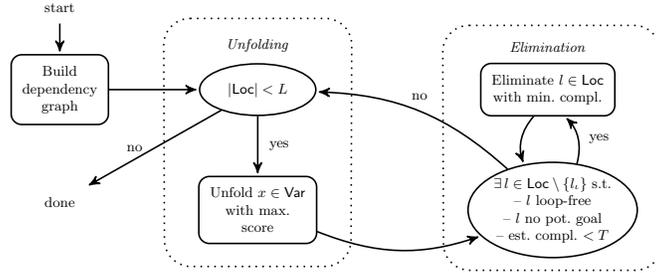

\section{Experiments}
\label{sec:experiments}

In this section, we report on our experimental evaluation of the implementation described in the previous section.

\begin{table}[t]
    \centering
    \caption{
        Reductions achieved by our control-flow elimination.
        Times are in ms.
    }
    \label{table:results}
    \begin{adjustbox}{max width=\textwidth}
        {\renewcommand{\arraystretch}{1.0}
            \setlength{\tabcolsep}{4pt}   
            \begin{tabular}{l l c  r   r  r  r  r  r  r  r  r  r r r}
                \toprule
                \multirow{2}{*}{Name} &  \multirow{2}{*}{Type} & Prop. & Red.  & \multirow{2}{*}{Params.} & \multicolumn{2}{c}{States} & \multicolumn{2}{c}{Transitions} & \multicolumn{2}{c}{Build time} & \multicolumn{2}{c}{Check time} & \multicolumn{2}{c}{Total time}  \\
                & & type & time  & &  orig. & red. & orig. & red. & orig. & red. & orig. & red. & orig. & red.  \\ \midrule
\multirow{4}{*}{\benchmark{brp}} & \multirow{4}{*}{dtmc} & \multirow{4}{*}{\benchmark{P}} & \multirow{4}{*}{\benchmark{134}}  & $2^{10}$/5 & 78.9\texttt{K} & -44\% & 106\texttt{K} & -33\% & 261 & -33\% & 22 & -38\% & \multirow{4}{*}{\benchmark{16,418}} & \multirow{4}{*}{\benchmark{-46\%}} \\
& &  &  & $2^{11}$/10 & 291\texttt{K} & -45\% & 397\texttt{K} & -33\% & 1,027 & -39\% & 101 & -46\% & & \\
& &  &  & $2^{12}$/20 & 1.11\texttt{M} & -46\% & 1.53\texttt{M} & -33\% & 3,945 & -48\% & 462 & -48\% & & \\
& &  &  & $2^{13}$/25 & 2.76\texttt{M} & -46\% & 3.8\texttt{M} & -33\% & 9,413 & -47\% & 1,187 & -47\% & & \\ \midrule
\multirow{1}{*}{\benchmark{coingame}} & \multirow{1}{*}{dtmc} & \multirow{1}{*}{\benchmark{P}} & \multirow{1}{*}{\benchmark{35}}  & $10^4$ & 20\texttt{K} & -50\% & 40\texttt{K} & -50\% & 53 & -24\% & 18,500 & -79\% & \multirow{1}{*}{\benchmark{18,553}} & \multirow{1}{*}{\benchmark{-78\%}} \\\midrule 
\multirow{1}{*}{\benchmark{dice5}} & \multirow{1}{*}{mdp} & \multirow{1}{*}{\benchmark{P}} & \multirow{1}{*}{\benchmark{671}}  & n/a & 371\texttt{K} & -84\% & 2.01\texttt{M} & -83\% & 1,709 & -82\% & 9,538 & -99\% & \multirow{1}{*}{\benchmark{11,247}} & \multirow{1}{*}{\benchmark{-91\%}} \\\midrule 
\multirow{2}{*}{\benchmark{eajs}} & \multirow{2}{*}{mdp} & \multirow{2}{*}{\benchmark{R}} & \multirow{2}{*}{\benchmark{223}}  & $10^3$ & 194\texttt{K} & -28\% & 326\texttt{K} & -1\% & 1,242 & -43\% & 220 & -32\% & \multirow{2}{*}{\benchmark{18,397}} & \multirow{2}{*}{\benchmark{-42\%}} \\
& &  &  & $10^4$ & 2\texttt{M} & -28\% & 3.38\texttt{M} & -1\% & 13,154 & -46\% & 3,780 & -31\% & & \\ \midrule
\multirow{2}{*}{\benchmark{grid}} & \multirow{2}{*}{dtmc} & \multirow{2}{*}{\benchmark{P}} & \multirow{2}{*}{\benchmark{117}}  & $10^4$ & 300\texttt{K} & -47\% & 410\texttt{K} & -34\% & 1,062 & -57\% & 17 & -52\% & \multirow{2}{*}{\benchmark{11,716}} & \multirow{2}{*}{\benchmark{-52\%}} \\
& &  &  & $10^5$ & 3\texttt{M} & -47\% & 4.1\texttt{M} & -34\% & 10,430 & -53\% & 207 & -54\% & & \\ \midrule
\multirow{1}{*}{\benchmark{hospital}} & \multirow{1}{*}{mdp} & \multirow{1}{*}{\benchmark{P}} & \multirow{1}{*}{\benchmark{57}}  & n/a & 160\texttt{K} & -66\% & 396\texttt{K} & -27\% & 502 & -50\% & 19 & -56\% & \multirow{1}{*}{\benchmark{521}} & \multirow{1}{*}{\benchmark{-39\%}} \\\midrule 
\multirow{4}{*}{\benchmark{nand}} & \multirow{4}{*}{dtmc} & \multirow{4}{*}{\benchmark{P}} & \multirow{4}{*}{\benchmark{80}}  & 20/4 & 308\texttt{K} & -79\% & 476\texttt{K} & -52\% & 589 & -45\% & 108 & -75\% & \multirow{4}{*}{\benchmark{86,060}} & \multirow{4}{*}{\benchmark{-56\%}} \\
& &  &  & 40/4 & 4\texttt{M} & -80\% & 6.29\texttt{M} & -51\% & 8,248 & -50\% & 1,859 & -77\% & & \\
& &  &  & 60/2 & 9.42\texttt{M} & -80\% & 14.9\texttt{M} & -50\% & 19,701 & -49\% & 4,685 & -76\% & & \\
& &  &  & 60/4 & 18.8\texttt{M} & -80\% & 29.8\texttt{M} & -50\% & 40,168 & -53\% & 10,703 & -77\% & & \\ \midrule
\multirow{4}{*}{\benchmark{nd-nand}} & \multirow{4}{*}{mdp} & \multirow{4}{*}{\benchmark{P}} & \multirow{4}{*}{\benchmark{106}} & 20/4 & 308\texttt{K} & -79\% & 476\texttt{K} & -52\% & 618 & -36\% & 127 & -74\% & \multirow{4}{*}{\benchmark{96,956}} & \multirow{4}{*}{\benchmark{-52\%}} \\
& &  &  & 40/4 & 4\texttt{M} & -80\% & 6.29\texttt{M} & -51\% & 8,783 & -42\% & 2,270 & -77\% & & \\
& &  &  & 60/2 & 9.42\texttt{M} & -80\% & 14.9\texttt{M} & -50\% & 21,792 & -47\% & 5,646 & -75\% & & \\
& &  &  & 60/4 & 18.8\texttt{M} & -80\% & 29.8\texttt{M} & -50\% & 44,409 & -46\% & 13,312 & -76\% & & \\ \midrule
\multirow{2}{*}{\benchmark{negotiation}} & \multirow{2}{*}{dtmc} & \multirow{2}{*}{\benchmark{P}} & \multirow{2}{*}{\benchmark{148}} & $10^4$ & 129\texttt{K} & -32\% & 184\texttt{K} & -26\% & 481 & -39\% & 22 & -49\% & \multirow{2}{*}{\benchmark{5,631}} & \multirow{2}{*}{\benchmark{-39\%}} \\
& &  &  & $10^5$ & 1.29\texttt{M} & -32\% & 1.84\texttt{M} & -26\% & 4,930 & -43\% & 197 & -30\% & & \\ \midrule
\multirow{2}{*}{\benchmark{pole}} & \multirow{2}{*}{dtmc} & \multirow{2}{*}{\benchmark{R}} & \multirow{2}{*}{\benchmark{208}} & $10^2$ & 315\texttt{K} & -46\% & 790\texttt{K} & -4\% & 1,496 & -46\% & 26 & -42\% & \multirow{2}{*}{\benchmark{17,431}} & \multirow{2}{*}{\benchmark{-45\%}} \\
& &  &  & $10^3$ & 3.16\texttt{M} & -46\% & 7.9\texttt{M} & -4\% & 15,503 & -47\% & 406 & -33\% & & \\
                \bottomrule
            \end{tabular}
        }
    \end{adjustbox}
\end{table}

\paragraph{Benchmarks.}
We have compiled a set of 10 control-flow intensive DTMC and MDP benchmarks from the literature.
Each benchmark model is equipped with a reachability or expected reward property.

\benchmark{brp} models a bounded retransmision protocol and is taken from the \prism\ benchmark suite.
\benchmark{coingame} is our running example from \Cref{fig:coingame_prism}.
\benchmark{dice5} is an example shipped with \storm\ and models rolling several dice, five in this case, that are themselves simulated by coinflips in parallel.
\benchmark{eajs} models energy-aware job scheduling and was first presented in~\cite{eajs}.
\benchmark{grid} is taken from~\cite{paynt} and represents a robot moving in a partially observable grid world.
\benchmark{hospital} is adapted from~\cite{hospital} and models a hospital inventory management problem.
\benchmark{nand} is the von Neumann NAND multiplexing system mentioned near the end of \Cref{sec:example}.
\benchmark{nd-nand} is a custom-made adaption of \benchmark{nand} where some probabilistic behavior has been replaced by non-determinism.
\benchmark{negotiation} is an adaption of the Alternating Offers Protocol from~\cite{negotiation} which is also included in the \prism\ case studies.
\benchmark{pole} is also from~\cite{paynt} and models balancing a pole in a noisy and unknown environment.
The problems \benchmark{brp}, \benchmark{eajs}, and \benchmark{nand} are part of the QComp benchmark set~\cite{qcomp}.

For all examples except \benchmark{dice5}, we have first flattened parallel compositions (if there were any) into a single module, cf.\ \Cref{sec:implementation}.


\paragraph{Setup.}
We report on two experiments.
In the first one, we compare the number of states and transitions as well as the model build and check times of the original and the reduced program (columns `States', `Transitions', `Build time', and `Check time' of \Cref{table:results}).
We work with \storm's default settings\footnote{By default, \storm\ builds the Markov model as a sparse graph data structure and uses (inexact) floating point arithmetic.}.
We also report the time needed for the reduction itself, including the time consumed by flattening (column `Red.\ time').
We always use the default configuration for our heuristics, i.e., \emph{we do not manually fine-tune} the heuristics for each benchmark.
We report on some additional experimental results obtained with fine-tuned heuristics in \iftoggle{arxiv}{\Cref{app:additional-exp}}{\cite{arxiv}}.
For the benchmarks where this is applicable, we consider the different parameter configurations given in column `Params.'.
Recall that in these cases, we need to compute the reduced program only once.
We report the amortized runtime of \storm\ on all parameter configurations vs.\ the runtime on the reduced models, including the time needed for reduction in the rightmost column `Total time'.
In the second, less extensive experiment, we compare our reductions to bisimulation minimization (\Cref{table:resultsBisim} below).
All experiments were conducted on a notebook with a 2.4GHz Quad-Core Intel Core i5 processor and 16GB of RAM.
The script for creating the table is available\footnote{\url{https://doi.org/10.5281/zenodo.5497947}}.

\paragraph{Results.}
Our default heuristics was able to reduce all considered models in terms of states (by 28-84\%) and transitions (by 1-83\%).
The total time for building and checking these models was decreased by  39-91\%.
The relative decrease in the number of states is usually more striking than the decrease in the number of transitions.
This is because, as explained in \Cref{sec:main}, location elimination always removes states but may add more commands to the PCFP and hence more transitions to the underlying Markov model.
Similarly, the time savings for model checking are often higher than the ones for model building; here, this is mostly because building our reduced model introduces some overhead due to the additional commands.
The reduction itself was always completed within a fraction of a second and is independent of the size of the underlying state space.

\paragraph{Bisimulation and control-flow reduction.}
In \Cref{table:resultsBisim}, we compare the compression achieved by \storm's probabilistic bisimulation engine, our method and \emph{both} techniques combined.
We also include the total time needed for reduction, model building and checking.
For the comparison, we have selected three benchmarks representing three different situations: (1) for \benchmark{brp}, the two techniques achieve similar reductions, (2) for \benchmark{nand}, our reduced model is smaller than the bisimulation quotient, and (3) for \benchmark{pole}, the situation is the other way around, i.e., the bisimulation quotient is (much) smaller than our reduced model.
Interestingly, combining the two techniques yields an even smaller model in all three cases.
This demonstrates the fact that \emph{control-flow reduction and bisimulation are orthogonal} to each other.
In the examples, control-flow reduction was also faster than bisimulation as the latter has to process large explicit state spaces.
It is thus an interesting direction for future work to combine program-level reduction techniques that yield bisimilar models with control-flow reduction.

\begin{table}[t]
    \centering
    \caption{
        Comparison of bisimulation minimization and our control-flow reduction (`CFR').
        Column `Total time' includes building, reducing and checking the model.
    }
    \label{table:resultsBisim}
    \begin{adjustbox}{max width=\textwidth}
        {\renewcommand{\arraystretch}{0.95}
            \setlength{\tabcolsep}{4pt}   
            \begin{tabular}{l  r  r  r  r  r  r r r r r}
                \toprule
                \multirow{2}{*}{Name} & \multirow{2}{*}{Params.} & \multicolumn{3}{c}{States} & \multicolumn{3}{c}{Transitions} & \multicolumn{3}{c}{Total time}  \\
                &  &  Bisim. & CFR & both & Bisim. & CFR & both & Bisim. & CFR & both\\ \midrule
                \benchmark{brp} & $2^{12}$/20 &                     598\texttt{K} & 606\texttt{K} & 344\texttt{K} &                     852\texttt{K} & 1.02\texttt{M} & 598\texttt{K} &                     4,767 & 2,883 & 2,965 \\ \midrule 
                \benchmark{nand} & 40/4 &                     3.21\texttt{M} & 816\texttt{K} & 678\texttt{K} &                     5\texttt{M} & 3.1\texttt{M} & 2.46\texttt{M} &                     17,868 & 5,588 & 8,199 \\ \midrule 
                \benchmark{pole} & $10^3$ &                     4.06\texttt{K} & 1.72\texttt{M} & 1.2\texttt{K} &                     12.2\texttt{K} & 7.54\texttt{M} & 9.82\texttt{K} &                     19,443 & 10,305 & 10,801 \\ 
                \bottomrule
            \end{tabular}
        }
    \end{adjustbox}
\end{table}

\paragraph{When does control-flow reduction work well?}
Our technique works best for models that use one or more explicit or implicit program counters.
Such program counters often come in form of a variable that determines which commands are currently available and that is updated after most execution steps.
Unfolding such variables typically yields several loop-free locations.
For example, the variable $\mathtt{f}$ in \Cref{fig:coingame_prism} is of this kind.
However, we again stress that there is no formal difference between program counter variables and ``data variables'' in our framework.
The distinction is made automatically by our heuristics; no additional user input is required.
Control-flow reduction yields especially good results if it can be applied compositionally such as in the \benchmark{dice5} benchmark.

\paragraph{Limitations.}
Finally, we remark that our approach is less applicable to extensively synchronizing parallel compositions of more than just a handful of modules.
The flattening approach then typically yields large PCFPs which are not well suited for symbolic techniques such as ours.
Larger PCFPs also require a significantly higher model building time.
Another limiting factor are dense variable dependencies in the sense of \Cref{sec:unfolding}, i.e., the variable dependency graph has relatively large BSCCs.
The latter, however, seems to rarely occur in practice.

\section{Conclusion}
\label{sec:conclusion}

This paper presented a property-directed ``unfold and eliminate'' technique on probabilistic control-flow programs which is applicable to state-based high-level modeling languages.
It preserves reachability probabilities and expected rewards exactly and can be used as a simplification front-end for any probabilistic model checker.
It can also handle parametric DTMC and MDP models where some key quantities are left open.
On existing benchmarks, our implementation achieved model compressions of up to an order of magnitude, even on models that have much larger bisimulation quotients.
Future work is to amend this approach to continuous-time models like CMTCs and Markov automata, and to further properties such as LTL.


%
%
%
 \bibliographystyle{splncs04}
 \bibliography{references}

\iftoggle{arxiv}{
    \clearpage
    \appendix
    \section{Transition Elimination: The General Case}
\label{app:general-elim}

We show how ingoing transition can be eliminated in the general case.
This generalizes the case of binary PCFPs covered by \Cref{fig:trans_elim}.

Let $\hat{l} \xrightarrow{\guard \to p:\update} l$ be an ingoing transition of $l \in \Loc$ that we wish to eliminate.
Suppose that the transition is part of a command
\[
	\gamma ~:= \qquad \hat{l},\, \guard \quad\to\quad p \colon  \update \colon l ~+~ \sum_{j=1}^n p_j \colon \update_j \colon l_j
\]
where we use the sum notation in the obvious way (note that $n = 0$ is possible when $p = 1$).
Our goal is to remove $\gamma$ completely (and thus not only the above transition) and replace it by a collection of $m$ new commands $\gamma_i'$, where $m$ is the number of commands available at location $l$.
More specifically, let
\[
	\gamma_i ~:= \qquad l,\, \altguard_i \quad\to\quad \sum_{j=1}^{m_i} q_{ij} \colon \altupdate_{ij} \colon l_{ij} 
\]
for $1 \leq i \leq m$ be one of those $m$ commands.
For all $1 \leq i \leq m$, we define the following \emph{new} commands:
\begin{align*}
	\gamma_i' ~:= \qquad\hat{l}, \, \guard \wedge \wp(u, \altguard_i) \quad\to\quad \sum_{j=1}^n p_j \colon \update_j \colon l_j ~+~ \sum_{j=1}^{m_i} p\,q_{ij} \colon (\update\seq\altupdate_{ij}) \colon l_{ij}~.
\end{align*}
After this operation, there will be a total of at most $m-1$ additional commands and $(m-1)n + \sum_{i=1}^m m_i -1$ additional transitions.

\section{Full Proofs}
\label{app:full_proofs}

\subsection{Proof of \Cref{lemma:unfoldsemantics} (Unfolding preserves the MDP semantics)}

Let $\pcfp = \pcfpinit$, $\Unf(\pcfp, x) = \pcfpinitprime$, $\mdp_{\pcfp} = \mdpinit$ and $\mdp_{\Unf(\pcfp, x)} = (S', P', \Act', \iota')$. 

For the proof, we identify MDP states named $\langle \langle l, \val(x) \rangle, \val \rangle$ and $\langle l, (\val(x), \val) \rangle$ and denote the extension of a valuation as $(\val(x), \val') = \val$ if $\val'$ is the restriction of $\val$ to $\Var \setminus \{ x \}$.
Consequently,
\[
    S = \Loc \times \{\val \in \dom\}  = \Loc \times \dom(x) \times \{\val \in \dom' \} = \Loc' \times \{\val \in \dom' \} = S'~.
\]

In the following, we show that $P = P'$.
To this end, we regard $P$ and $P'$ as the \emph{sets} of transitions they describe.

We first show $P \subseteq P'$.
Let $\langle l, \val \rangle \xrightarrow{a_{\gamma}, p} \langle l', \update(\val) \rangle$ be a transition in $\mdp_\pcfp$.
Then, by \Cref{def:mdpsemantics}, we have $l \xrightarrow{\guard \rightarrow p : \update} l' \wedge \val \models \guard \wedge \update(\val) \in \dom$ in $\mathfrak P$.
We define $\val_x: \{x\} \to \dom(x)$ with $\val_x(x) = \val(x)$.
By the rule from \Cref{def:unfolding} it follows that
\[
    \langle\, l, \val_x(x) \,\rangle \xrightarrow{\guard[\val_x] \,\to\, p\,:\,\update[\val_x]} \langle\, l', \update(\val_x)(x) \,\rangle
\]
holds in $\Unf(\mathfrak P, x)$.
Let $\val'$ be the restriction of $\val$ to $\Var \setminus\{x\}$.
From $\val \models \guard$, it follows that $\val' \models \guard[\val_x]$.
Since $\update(\val) \in \dom$ holds in $\mathfrak P$, $u(\val') \in \dom'$ holds in $\Unf(\mathfrak P, x)$.
Applying the rule from \Cref{def:mdpsemantics} yields that
\[
    \langle\, \langle\, l, \val_x(x) \,\rangle, \val' \,\rangle \xrightarrow{a_{\gamma}, p} \langle\, \langle\, l', \update(\val_x)(x) \,\rangle, \update(\val) \,\rangle
\]
holds in $\Unf(\mathfrak P, x)$.
As $\langle \langle l, \val_x(x)\rangle, \val' \rangle =\langle l, \val \rangle$ and $\langle \langle l', u(\val_x)(x)\rangle, \update(\val') \rangle = \langle l', \update(\val) \rangle$, we have thus shown that $P \subseteq P'$.

We now show $P' \subseteq P$.
Let $\langle \langle l, \val_x(x)\rangle, \val' \rangle \xrightarrow{a_{\gamma}, p} \langle \langle l', u(\val_x)(x)\rangle, \update(\val') \rangle$ be a transition in $\Unf(\pcfp, x)$.
This implies that there exists a guard $\guard'$ over $\Var'$ such that
\[
    \langle\, l, \val_x(x) \,\rangle \xrightarrow{\guard' \rightarrow p \,:\, \update} \langle\, l', u(\val_x)(x) \,\rangle ~\wedge~ \val' \models \guard' ~\wedge~ \update(\val') \in \dom
\]
holds in $\Unf(\pcfp, x)$.
We have $u(\val_x)(x) \in \dom(x)$ because otherwise $\langle l', u(\val_x)(x) \rangle \in \Loc'$ would not hold.
We have already shown that $u(\val') \in \dom'$ holds, so for the combination $\val = (\val_x, \val')$, $\val \in \dom$ also holds.
As $\langle l, \val_x(x) \rangle \xrightarrow{\guard' \rightarrow p : \update} \langle l', u(\val_x)(x) \rangle$ is a transition in $\Unf(\pcfp, x)$, by \Cref{def:unfolding} there is a $\guard$ with $\guard' = \guard[\val_x]$ such that $l \xrightarrow{\guard \rightarrow p : \update} l'$ holds in $\pcfp$.
Additionally, because $\val' \models \guard'$ and $\guard' = \guard[\val_x]$, it holds that $\val \models \guard$.
By the rule from \Cref{def:mdpsemantics}, $\langle l, \val \rangle \xrightarrow{a_{\gamma}, p} \langle l', \update(\val) \rangle$ must therefore be a transition in $\mdp_{\pcfp}$.

Finally, notice that $\iota = \langle l_\iota, \val_\iota \rangle = \langle l_\iota, \val_\iota(x), \val'_\iota\rangle = \langle \langle l_\iota, \val_\iota(x) \rangle, \val'_\iota \rangle = \iota'$.

\qedProof

\subsection{Proof of \Cref{lem:trans_elim_correct} (Transition elimination is correct)}
\label{proof:trans_elim_correct}

\newcommand{\conf}[2]{\langle\, #1, #2 \,\rangle}

We formally prove that the transition elimination rule from \Cref{fig:trans_elim} is correct in the sense of \Cref{lem:trans_elim_correct}, i.e., that it yields a reachability equivalent PCFP.

First, we argue that state elimination in Markov chains (\Cref{fig:mc_state_elim}) can also be extended to MDPs.
We explain the corresponding rule for binary MDPs (each state has two available actions and each action leads to at most 2 two distinct successor states; the generalization to arbitrary MDPs is straightforward).
Let $\mdp = \mdpinit$ be a (binary) MDP.
Let $s_1 \in S$ and suppose we want to eliminate the ingoing transition $\mdptrans{s}{\alpha}{p_1}{s_1}$ (see \Cref{fig:mdp-elim}).
To this end, we also remove the other transition $\mdptrans{s}{\alpha}{p_2}{s_2}$, and introduce the new action labels $\alpha;\gamma$ and $\alpha;\delta$ and the following new transitions:
\begin{align*}
    &\mdptrans{s}{\alpha;\gamma}{p_2}{s_2}
    \qquad \mdptrans{s}{\alpha;\gamma}{p_1q_{11}}{s_{11}}
    \qquad \mdptrans{s}{\alpha;\gamma}{p_1q_{12}}{s_{12}}\\
    &\mdptrans{s}{\alpha;\delta}{p_2}{s_2}
    \qquad \mdptrans{s}{\alpha;\delta}{p_1q_{21}}{s_{21}}
    \qquad \mdptrans{s}{\alpha;\delta}{p_1q_{22}}{s_{22}} ~.
\end{align*}
Let $G \subseteq S$ and let $\mdp'$ be the MDP resulting from $\mdp$ by applying the above transformation.
We claim that $\max_{\sigma}\prob_{\mdp^\sigma}(\reach G) = \max_{\sigma} \prob_{\mdp'^\sigma}(\reach G)$ provided that $s_1 \notin G$ and $s_1 \neq \sinit$, i.e., $s_1$ is neither an initial state nor a goal state (the proof is analogous for minimal reachability probabilities).
Here, the maximum ranges over all memoryless and deterministic schedulers of $\mdp$ and $\mdp'$, respectively.
Recall from \Cref{sec:prelims} and \cite{puterman} that this class of schedulers suffices for maximal and minimal reachability probabilities.
To show our claim let $\sigma$ be a maximizing scheduler in $\mdp$.
W.l.o.g.\ assume that $\sigma(s) = \alpha$ since otherwise there is nothing to show because the transition we eliminate is never actually chosen.
Further, suppose that $\sigma(s_1) = \gamma$ (the other case $\sigma(s_1) = \delta$ is symmetric).
We now define a scheduler $\sigma'$ in $\mdp'$ that is like $\sigma$ but selects action $\sigma'(s) = \alpha;\gamma$.
Then it is easy to see that the induced Markov chain $\mdp'^{\sigma'}$ is obtained from $\mdp^\sigma$ by applying Markov chain transition elimination to the transition $s \xrightarrow{p_1} s_1$.
Since $s_1$ was neither initial nor contained in the goal set $G$, this preserves reachability probabilities w.r.t.\ $G$.

Now let $\pcfp = \pcfpinit$ be a (binary, well-formed) PCFP.
Let $l_1 \in \Loc \setminus \{l_\iota\}$ be no potential goal with respect to goal predicate $\goalpred$ over $\Var$ and let $\pcfp'$ be obtained from $\pcfp$ by eliminating transition $l \xrightarrow{\guard \to p_1:\update_1} l_1$ according to the rule in \Cref{fig:trans_elim}.
We have to show that $\pcfp$ and $\pcfp'$ are reachability equivalent w.r.t.\ $\goalpred$.
To this end, we show that the semantic MDP $\mdp' := \mdp_{\pcfp'}$ can be obtained from $\mdp := \mdp_{\pcfp}$ by applying MDP transition elimination as in \Cref{fig:mdp-elim} repeatedly.

By \Cref{def:mdpsemantics} we have for all $\val \models \guard$, $\val' \models \altguard_1$ and $\val'' \models \altguard_2$ the following transitions $(\star)$ in the MDP $\mdp$:
\begin{align*}
    &\mdptrans{\conf{l}{\val}}{a_{\alpha}}{p_1}{\conf{l_1}{\update_1(\val)}} \\
    &\mdptrans{\conf{l}{\val}}{a_{\alpha}}{p_2}{\conf{l_2}{\update_2(\val)}} \\
    &\mdptrans{\conf{l_1}{\val'}}{a_{\gamma}}{q_{11}}{\conf{l_{11}}{v_{11}(\val')}} \tag{$\star$}\\
    &\mdptrans{\conf{l_1}{\val'}}{a_{\gamma}}{q_{12}}{\conf{l_{12}}{v_{12}(\val')}} \\
    &\mdptrans{\conf{l_1}{\val''}}{a_{\delta}}{q_{21}}{\conf{l_{21}}{v_{21}(\val'')}} \\
    &\mdptrans{\conf{l_1}{\val''}}{a_{\delta}}{q_{22}}{\conf{l_{22}}{v_{22}(\val'')}}~.
\end{align*}
Here, $\alpha$, $\beta$, and $\delta$ refer to the commands containing transitions $l \xrightarrow{\guard \to p_1:\update_1} l_1$, $l_1 \xrightarrow{\altguard_1 \to q_{11}:v_{11}} l_{11}$, and $l_1 \xrightarrow{\altguard_2 \to q_{21}:v_{21}} l_{21}$, respectively.

On the other hand, again by \Cref{def:mdpsemantics}, in $\mdp'$ for all $\val \models \guard \land \wp(\update_1, \altguard_1)$ and $\val' \models \guard \land \wp(\update_1, \altguard_2)$ there exist the transitions $(\star\star)$
\begin{align*}
    &\mdptrans{\conf{l}{\val}}{a_{\alpha;\gamma}}{p_1q_{11}}{\conf{l_{11}}{(\update_1 \seq v_{11})(\val)}} \\
    &\mdptrans{\conf{l}{\val}}{a_{\alpha;\gamma}}{p_1q_{12}}{\conf{l_{12}}{(\update_1 \seq v_{12})(\val)}} \\
    &\mdptrans{\conf{l}{\val}}{a_{\alpha;\gamma}}{p_2}{\conf{l_2}{\update_2 (\val)}} \tag{$\star\star$}\\
    &\mdptrans{\conf{l}{\val'}}{a_{\alpha;\delta}}{p_1q_{21}}{\conf{l_{21}}{(\update_1 \seq v_{21})(\val')}} \\
    &\mdptrans{\conf{l}{\val'}}{a_{\alpha;\delta}}{p_1q_{22}}{\conf{l_{22}}{(\update_1 \seq v_{22})(\val')}} \\
    &\mdptrans{\conf{l}{\val'}}{a_{\alpha;\delta}}{p_2}{\conf{l_2}{\update_2 (\val')}} ~.
\end{align*}
Here, $\alpha;\gamma$ and $\alpha;\delta$ denote the two commands available at $l$ in $\pcfp'$.
We claim that all the transitions $(\star\star)$ in $\mdp'$ are obtained from the transitions $(\star)$ in $\mdp$ by applying the rule in \Cref{fig:mdp-elim}.
To see this, note that the condition $\val \models \guard \land \wp(\update_1, \altguard_1)$ holds iff $\val \models \guard$ and $\update_1(\val) \models \altguard_1$.
Therefore, for all $\val$ and $ i \in \{1,2\}$ we have that
\begin{align*}
    &\mdptrans{\conf{l}{\val}}{a_{\alpha;\gamma}}{p_1q_{1i}}{\conf{l_{1i}}{(\update_1 \seq v_{1i})(\val)}} \tag{is a transition in $\mdp'$} \\
    \iff & \mdptrans{\conf{l}{\val}}{a_{\alpha}}{p_1}{\conf{l_1}{\update_1(\val)}} \\
    \text{and } & \mdptrans{\conf{l_1}{\update_1(\val)}}{a_{\gamma}}{q_{1i}}{\conf{l_{1i}}{v_{1i}(\update_1(\val))}} \tag{are transitions in $\mdp$} ~.
\end{align*}
Similarly, the condition $\val \models \guard \land \wp(\update_1, \altguard_2)$ holds iff $\val \models \guard$ and $\update_1(\val) \models \altguard_2$.
Therefore, for all $\val$ and $ i \in \{1,2\}$ we have that
\begin{align*}
    &\mdptrans{\conf{l}{\val'}}{a_{\alpha;\delta}}{p_1q_{2i}}{\conf{l_{2i}}{(\update_1 \seq v_{2i})(\val')}} \tag{is a transition in $\mdp'$} \\
    \iff & \mdptrans{\conf{l}{\val}}{a_{\alpha}}{p_1}{\conf{l_1}{\update_1(\val)}} \\
    \text{and } & \mdptrans{\conf{l_1}{\update_1(\val)}}{a_{\gamma}}{q_{2i}}{\conf{l_{2i}}{v_{2i}(\update_1(\val))}} \tag{are transitions in $\mdp$} ~.
\end{align*}
We now show that the transitions $(\star\star)$ in $\mdp'$ can be constructed from the transitions $(\star)$ by MDP transition elimination.
To this end, we make the following case distinction for all $\val \in \dom$:
\begin{itemize}
    \item  $\val \models \guard \land \wp(\update_1, \altguard_1)$ and $\val \models \guard \land \wp(\update_1, \altguard_2)$.
    In this case, by the previous observation we can apply the MDP transition elimination rule as in \Cref{fig:mdp-elim} with
    \begin{itemize}
        \item $s = \conf{l}{\val}$,
        \item $s_1 = \conf{l_1}{\update_1(\val)}$, $s_2 = \conf{l_2}{\update_2(\val)}$,
        \item $s_{11} = \conf{l_{11}}{v_{11}(\update_1(\val))}$, $s_{12} = \conf{l_{12}}{v_{12}(\update_1(\val))}$
        \item $s_{21} = \conf{l_{21}}{v_{21}(\update_1(\val))}$, $s_{22} = \conf{l_{22}}{v_{22}(\update_1(\val))}$, and
        \item $\alpha = a_{\alpha}$, $\gamma = a_{\gamma}$, $\delta = a_{\delta}$
    \end{itemize}
    to obtain the desired transitions $\mdptrans{\conf{l}{\val}}{a_{\alpha;\gamma}}{p_1q_{1i}}{\conf{l_{1i}}{(\update_1 \seq v_{1i})(\val)}}$ and $\mdptrans{\conf{l}{\val'}}{a_{\alpha;\delta}}{p_1q_{2i}}{\conf{l_{2i}}{(\update_1 \seq v_{2i})(\val')}}$, $i \in \{1,2\}$ in $\mdp'$.
    \item  $\val \models \guard \land \wp(\update_1, \altguard_1)$ and $\val \nvDash \guard \land \wp(\update_1, \altguard_2)$.
    In this case, we apply the MDP transition elimination rule just ``partially'' as follows:
    \begin{itemize}
        \item $s = \conf{l}{\val}$,
        \item $s_1 = \conf{l_1}{\update_1(\val)}$, $s_2 = \conf{l_2}{\update_2(\val)}$,
        \item $s_{11} = \conf{l_{11}}{v_{11}(\update_1(\val))}$, $s_{12} = \conf{l_{12}}{v_{12}(\update_1(\val))}$,  and
        \item $\alpha = a_{\alpha}$, $\gamma = a_{\gamma}$,
    \end{itemize}
    i.e., at $s_2$ there is just a single action $\gamma$ available, but the other action $\delta$ as in \Cref{fig:mdp-elim} is not present.
    \item  The case $\val \models \guard \land \wp(\update_1, \altguard_2)$ and $\val \nvDash \guard \land \wp(\update_1, \altguard_1)$ is symmetric to the previous case.
    \item For the remaining case $\val \nvDash \guard \land \wp(\update_1, \altguard_1)$ and $\val \nvDash \guard \land \wp(\update_1, \altguard_2)$ there is nothing to show.
\end{itemize}

Overall, the claim follows as the MDP transition elimination rule preserves reachability probabilities and because $\conf{l_1}{\val} \notin G_{\goalpred}$ for all $\val \in \dom$ as $l_1$ is no potential goal w.r.t.\ $\goalpred$ by assumption.

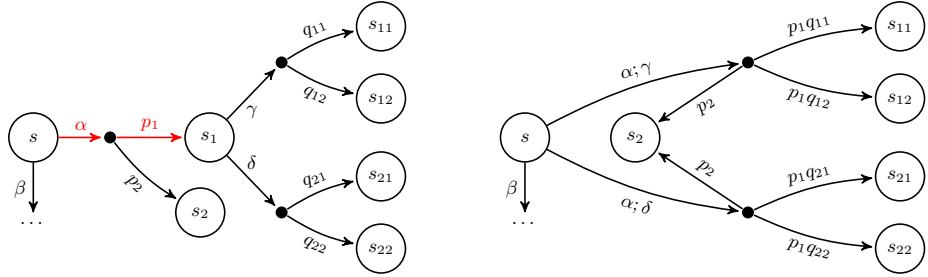
\begin{figure}[t]
	\centering
	\begin{tikzpicture}[myArrowStyle, node distance = 7mm and 9mm, every node/.style={scale=0.8}]
		\node[state] (lhat) {$s$};
		\node[below =of lhat] (dots) {$\dots$};
		\node[circle, inner sep=2pt, right=6mm of lhat, fill=black] (phi1dot) {};
		\node[state, right=of phi1dot] (l1) {$s_1$};
		\node[state, below right=of phi1dot] (l2) {$s_2$};
		\node[circle, inner sep=2pt, above right=of l1, xshift=-3mm, fill=black] (psi1dot) {};
		\node[circle, inner sep=2pt, below right=of l1, xshift=-3mm, fill=black] (psi2dot) {};
		\node[state,  right=of psi1dot, yshift=6mm] (l11) {$s_{11}$};
		\node[state,  right=of psi1dot, yshift=-6mm] (l12) {$s_{12}$};
		\node[state,  right=of psi2dot, yshift=6mm] (l21) {$s_{21}$};
		\node[state,  right=of psi2dot, yshift=-6mm] (l22) {$s_{22}$};
		
		\draw[->,red] (lhat) edge node[above] {$\alpha$} (phi1dot);
		\draw[->] (lhat) edge node[left] {$\beta$} (dots);
		
		\draw[->,red] (phi1dot) edge node[above] {$p_1$} (l1);
		\draw[->] (phi1dot) edge[bend right=10] node[below, sloped] {$p_2$} (l2);
		
		\draw[->] (l1) edge node[below,inner sep=2mm] {$\gamma$} (psi1dot);
		\draw[->] (l1) edge node[above,inner sep=2mm] {$\delta$} (psi2dot);
		
		\draw[->] (psi1dot) edge[bend left=10] node[above,sloped] {$q_{11}$} (l11);
		\draw[->] (psi1dot) edge[bend right=10] node[below,sloped] {$q_{12}$} (l12);
		\draw[->] (psi2dot) edge[bend left=10] node[above,sloped] {$q_{21}$} (l21);
		\draw[->] (psi2dot) edge[bend right=10] node[below,sloped] {$q_{22}$} (l22);
	\end{tikzpicture}
	\hfill
	\begin{tikzpicture}[myArrowStyle, node distance = 7mm and 12mm, , every node/.style={scale=0.8}]
		\node[state] (lhat) {$s$};
		\node[below =of lhat] (dots) {$\dots$};
		
		\node[state, right=8mm of lhat] (l2) {$s_2$};
	
		\node[circle, inner sep=2pt, above right=of l2, fill=black] (psi1dot) {};
		\node[circle, inner sep=2pt, below right=of l2, fill=black] (psi2dot) {};
		\node[state,  right=16mm of psi1dot, yshift=6mm] (l11) {$s_{11}$};
		\node[state,  right=16mm of psi1dot, yshift=-6mm] (l12) {$s_{12}$};
		\node[state,  right=16mm of psi2dot, yshift=6mm] (l21) {$s_{21}$};
		\node[state,  right=16mm of psi2dot, yshift=-6mm] (l22) {$s_{22}$};
		
		\draw[->] (lhat) edge node[left] {$\beta$} (dots);
		
		\draw[->] (lhat) edge[bend left=10] node[above,sloped] {$\alpha;\gamma$} (psi1dot);
		\draw[->] (lhat) edge[bend right=10] node[below,sloped] {$\alpha;\delta$} (psi2dot);
		
		\draw[->] (psi1dot) edge node[below,sloped] {$p_2$} (l2);
		\draw[->] (psi2dot) edge node[above,sloped] {$p_2$} (l2);

		\draw[->] (psi1dot) edge[bend left=10] node[above,sloped] {$p_1 q_{11}$} (l11);
		\draw[->] (psi1dot) edge[bend right=10] node[below,sloped] {$p_1 q_{12}$} (l12);
		\draw[->] (psi2dot) edge[bend left=10] node[above,sloped] {$p_1 q_{21}$} (l21);
		\draw[->] (psi2dot) edge[bend right=10] node[below,sloped] {$p_1 q_{22}$} (l22);
	\end{tikzpicture}
	\caption{Transition elimination in binary MDPs. The rule preserves reachability probabilities provided that $s_1$ is neither initial nor a goal state. The transformation also works if there is just one action available at $s_1$.}
	\label{fig:mdp-elim}
\end{figure}

\qedProof

\subsection{Proof of \Cref{thm:complexity} (Complexity of Location Elimination)}
\label{proof:complexity}

We restate the theorem for convenience:

\thmcomplexity*

\noindent We now prove the theorem by discussing each item individually:
\begin{itemize}
    \item 
    Let $\gamma_1,\ldots,\gamma_n$ be the $n$ distinct commands in $\Cmd$ that have a transition leading to $l$, i.e., for $i = 1,\ldots,n$ we have at least one transition of the form
    \[
        \pcfptrans{l_i}{\guard_i}{p_i}{\update_i}{l}
    \]
    contained in $\gamma_i$ .
    Moreover, the multiplicity of these transitions is at most $m$ by assumption.
    We \emph{process} each of these commands as follows:
    We apply transition elimination to an arbitrary $\gamma_1$-transition first and then to all \emph{new} ingoing transitions to $l$ created by this\footnote{See the paragraph above \Cref{thm:locelim} for an explanation why transition elimination may create new ingoing transitions; recall that this is impossible in the case of Markov chain transition elimination. Also recall that new transitions may only be created if the multiplicity of the transition to be eliminated is greater than 1.}.
    We iterate this until no new ingoing transitions are created.
    After that, we process the other commands $\gamma_2,\ldots,\gamma_n$.
    
    The number $I(m)$ of iterations of the algorithm in \Cref{thm:locelim} (which is equal to the number of times we apply transition elimination) for processing a single command with transitions leading to $l$ with multiplicity $m$  satisfies the recurrence $I(1)=1$ and $I(m)= 1 + kI(m-1)$ for all $m > 1$ since eliminating a transition with multiplicity $m > 1$ yields $k$ new commands with multiplicity $m-1$ each.
    The solution of this recurrence is $I(m) = \sum_{i=0}^{m-1} k^i = (k^m {-}1)/(k{-}1)$.
    Thus, to process all $n$ commands, $n(k^m {-}1)/(k{-}1)$ iterations suffice.
    \item
    As in the previous item, we process each command $\gamma_1,\ldots,\gamma_n$ one after another.
    We may think of the following stages when processing one such command:
    In stage 1, there is a the single command $\gamma_i$ with multiplicity $m$.
    In stage $j$ for $j > 1$, the commands from the previous stage are transformed into $k$ new commands with multiplicity $m-j+1$ each.
    In the final stage $m$, there are thus $k^{m-1}$ commands with multiplicity $1$ each.
    Eliminating all of them yields $k \cdot k^{m-1} = k^m$ new commands, but no new commands with a transition leading to $l$.
    Hence, the algorithm creates a total of at most $nk^m$ new commands after processing all $n$ commands $\gamma_1,\ldots,\gamma_n$.
    \item
    We first give an example for $n=1$ and $k=2$.
    Consider the deterministic PCFP $\pcfp$ with $\Var = \{x_i, y_i \mid 1 \leq i \leq m\}$ and $\dom(x_i) = \dom(y_i) = \{0,1\}$ depicted in \Cref{fig:exponential}.
    Let $\pcfp'$ be the result of applying the location elimination algorithm from \Cref{thm:locelim} to the loop-free location $l$ of $\pcfp$.
    As eliminating $l$ with our rule preserves reachability, in the Markov chain $\mdp_{\pcfp'}$ it holds that
    \[
    	P(\langle l', \val \rangle, \langle l_1, \vec{0} \rangle) = c \sum_{i} \delta_{1, \val(x_i)} \frac{1}{2^i}
    \]
    for all $\val \in \dom$ with $\val(y_i)=0$ for all $i$ and where $\delta$ is the Kronecker-Delta.
    Note that all these $2^m$ probabilities are pairwise distinct:
    In each case, the probability is equal to the binary decimal $0.\val(x_1)\ldots\val(x_m)$ multiplied by the normalizing constant $c$.
    
    This implies that $\pcfp'$ must have at least $2^m$ commands available at $l'$ as otherwise there would be at most $2^m - 1$ pairwise distinct probabilities $P(\langle l', \val \rangle, \langle l_1, \vec{0} \rangle)$ in $\mdp_{{\pcfp'}}$, where $\val \in \dom$ with $\val(y_i)=0$ for all $i=1,\ldots,m$.
    Moreover, it is clear that these commands have satisfiable guards and are pairwise distinct.
    
    The example can be extended to $n > 1$ as well:
    We simply make $n$ copies of the location $l'$.
    
    It is less obvious how to adapt the example to $k > 2$, and we shall content ourselves with the exponential lower bound that is already implied by the $k=2$ case.
\end{itemize}

\begin{figure}[h]
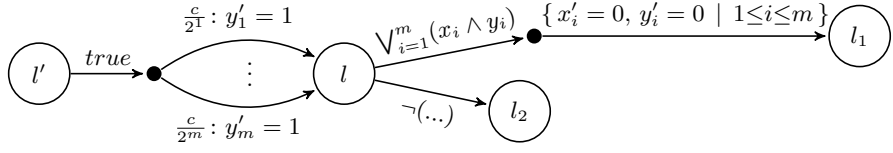

    \centering
    \figexponential
    \caption{The PCFP $\pcfp$ used in the proof of \Cref{thm:complexity}.
    The transitions from $l'$ to $l$ have multiplicity $m$ each.
    Variables $x,y$ have Boolean domain, $\wedge$ denotes logical conjunction, and $c$ is a normalizing constant.}
    \label{fig:exponentialapp}
\end{figure}

\qedProof

\section{Eliminating Self-Loops}
\label{app:selfloops}

Analyzing loops is notoriously difficult---even in non-probabilistic programs---and usually boils down to finding loop-invariants. The general idea of this paper is to fall back to further variable unfolding (\Cref{sec:unfolding}) if no location without self-loops exists. However, in several special cases, we can eliminate self-loops.

First, we observe that there are ``lucky cases'' where the transition elimination rule is sufficient even for locations with self-loops:
Assume that location $l$ has a self-loop and an ingoing transition from source location $\hat{l} \neq l$, to which we apply transition elimination. Our rule then yields (among others) transitions of the form
\[\hat{l} \,\xrightarrow{\guard \wedge \wp(u, \altguard) \,\to\, pq:\, \altupdate\seq\update}\, l\]
where $l$ remains a target due to its self-loop.
However, it is possible that the guards $\guard \wedge \wp(u, \altguard)$ in these transitions are all unsatisfiable.
Likewise, transition elimination can be applied directly to self-loops to eliminate similar lucky cases.

We develop one further loop elimination rule.
As a first observation, suppose that location $l$ has a self-loop $l \xrightarrow{\guard\to p:\update}l$ with $u = \nop$ (an effectless update) and $p < 1$.
Then the probability $p$ can be redistributed over the remaining probabilistic choices in the corresponding command by multiplying their respective probabilities with $(1-p)^{-1}$, just like in Markov chain self-loop elimination (\Cref{fig:mc_state_elim}, left).
Now suppose that $\update$ is \emph{idempotent}, that is $\update(\update(\val)) = \update(\val)$ for all variable valuations $\val$.
For instance, the update $x'=y$ is idempotent, but the update $x'=x+1$ is not.
In this case, we can also formulate an elimination rule:

\begin{proposition}
	\label{prop:idempotent_correct}
	Let $\pcfp$ be a PCFP with self-loop $l \xrightarrow{\guard\to p:\update}l$ where $\update$ is idempotent and $p < 1$.
    Further, suppose that $l$ is no potential goal w.r.t.\ goal predicate $\goalpred$ \emph{or} that $\wp(u_1, \goalpred)$ is unsat. Let $\pcfp'$ be the resulting PCFP after applying the rule in \Cref{fig:idempotent}. Then $\pcfp$ and $\pcfp'$ are $\goalpred$-reachability equivalent.
\end{proposition}
\begin{proof}
    This rule is seen to be correct by introducing a ``temporal location'' $\hat{l}$ with $l \xrightarrow{\guard\to p_1:\update_1} \hat{l}$ and that is otherwise like $l$ except $\hat{l} \xrightarrow{\guard\to p_1:\nop} \hat{l}$ is a $\nop$ self-loop.
    The intuitive meaning of $\hat{l}$ is that it encodes the state that results from applying $\update_1$ in location $l$.
    Since $\update_1$ is idempotent, $\hat{l}$ has a $\nop$ self-loop that can be immediately eliminated by redistributing the probability $p_1$ over the other choices.
    Applying the transition elimination rule (\Cref{lem:trans_elim_correct}) to $l \xrightarrow{\guard\to p_1 \update_1} \hat{l}$ then yields the result as claimed.
    \qedProof
\end{proof}
Notice though that this rule only effectively removes a self-loop if the four locations in \Cref{fig:idempotent} are pairwise distinct as it otherwise introduces new self-loops.

\begin{figure}[t]
	\centering
	\begin{tikzpicture}[myArrowStyle, every node/.style={scale=0.8}]
		\node[state] (l) {$l$};
		\node[circle, inner sep=2pt, left= 6mm of l, fill=black] (phidot) {};
		\node[circle, inner sep=2pt, right= 6mm of l, fill=black] (psidot) {};
		\node[state, left=of phidot] (l') {$l'$};
		\node[state, above=6 mm of psidot, xshift=10mm] (l1) {$l_1$};
		\node[state, below=6 mm of psidot, xshift=10mm] (l2) {$l_2$};
		
		\draw[->] (l) --node[below] {$\guard$} (phidot);
		\draw[->, red] (phidot) edge[bend left=45] node[above] {$p_1: \update_1$} (l);
		\draw[->] (phidot) -- node[above] {$p_2: \update_2$} (l');
		\draw[->] (l) --node[below] {$\altguard$} (psidot);
		\draw[->] (psidot) --node[above,sloped] {$q_1:v_1$} (l1);
		\draw[->] (psidot) --node[below,sloped] {$q_2:v_2$} (l2);
	
	\node[state, right= 55mm of l] (l) {$l$};
	\node[circle, inner sep=2pt, left= 6mm of l, fill=black] (phidot) {};
	\node[circle, inner sep=2pt, above= of phidot, fill=black] (phidotnew) {};
	\node[circle, inner sep=2pt, right= 6mm of l, fill=black] (psidot) {};
	\node[state, left=of phidot] (l') {$l'$};
	\node[state, above=6 mm of psidot, xshift=10mm] (l1) {$l_1$};
	\node[state, below=6 mm of psidot, xshift=10mm] (l2) {$l_2$};
	
	\draw[->] (l) --node[below, near start] {$\guard_1$} (phidot);
	\draw[->] (l) --node[right, near start] {$\guard_2$} (phidotnew);
	\draw[->] (phidotnew) edge[bend right] node[above, sloped] {$p_1: \update_1\seq \update_2$} (l');
	\draw[->] (phidotnew) edge[bend left] node[above, sloped] {$p_2: \update_2$} (l');
	\draw[->] (phidot) -- node[below] {$p_2: \update_2$} (l');
	\draw[->] (l) --node[below] {$\altguard$} (psidot);
	\draw[->] (psidot) --node[above,sloped] {$q_1:v_1$} (l1);
	\draw[->] (psidot) --node[below,sloped] {$q_2:v_2$} (l2);
	\draw[->] (phidot) edge[bend left] node[above, near end] {$p_1q_1:\update_1\seq v_1$} (l1);
	\draw[->] (phidot) edge[bend right] node[below,near end] {$p_1q_2:\update_1\seq v_2$} (l2);
	\end{tikzpicture}
	\caption{Elimination rule for idempotent self-loops. Self-loop $l \xrightarrow{p_1:\update_1} l$ is eliminated. In the figure, $\guard_1 := \guard \wedge \wp(\update_1, \altguard)$ and $\guard_2 := \guard \wedge \wp(\update_1, \guard)$}
	\label{fig:idempotent}
\end{figure}
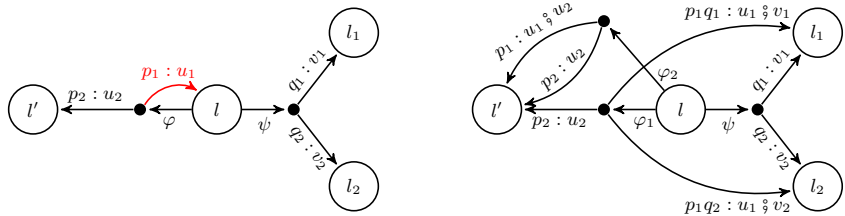

\section{Benchmarks Details}
\label{app:benchmarks}

In the following list, the ``short descriptions'' formatted as quotes are literal quotes from the previously listed references.

\begin{description}
    
    \item[\benchmark{brp}] Bounded retransmission protocol
    \begin{itemize}
        \item From: \cite{brp}, \prism\ benchmark suite, QComp~\cite{qcomp} benchmark set
        \item Short description: \emph{``The BRP protocol sends a file in a number of chunks, but allows only a bounded number of retransmissions of each chunk.''}
        \item Verified property: \texttt{P=? [F s=5]}
        \item Parameters: \texttt{N} (positive integer): number of chunks in a file, \texttt{MAX}: (positive integer) maximum number of retransmissions.
    \end{itemize}

    \item[\benchmark{coingame}] Coin game used as running example in this paper
    \begin{itemize}
        \item From: this paper
        \item Short description: See \Cref{sec:example}.
        \item Verified property: \texttt{ P=? [F (x>= N) \& (f=false) ]}
        \item Parameters: \texttt{N} number of rounds.
    \end{itemize}

     \item[\benchmark{dice5}] Rolling several dice in parallel
    \begin{itemize}
        \item From: Example shipped with \storm.
        \item Short description: This benchmark models rolling several dice, five in this case, in parallel. The individual dice are themselves simulated by coin flips similar to the Knuth-Yao die.
        \item Verified property: \texttt{Pmax=? [F s1=7 \& s2=7 \& s3=7 \& s4=7 \& s5=7 \& d1+d2+d3+d4+d5=15]}
        \item Parameters: n/a
    \end{itemize}

    \item[\benchmark{eajs}] Energy-aware job scheduling
    \begin{itemize}
        \item From: \cite{eajs}, QComp~\cite{qcomp} benchmark set
        \item Short description: \emph{``A system of N processes which need to enter a critical section in order to perform tasks, each within a given deadline. Access to the critical section is exclusively granted by a scheduler, which selects processes only if they have requested to enter.''}
        \item Verified property: \texttt{R\{"utilityLocal"\}min=? [F localFailure]}
        \item Parameters: \texttt{energy\_capacity}	(positive integer):	The amount of available energy.
    \end{itemize}

    \item[\benchmark{grid}] Partially observable grid world
    \begin{itemize}
        \item From: \cite{paynt}
        \item Short description: Models a robot moving in a partially observable grid world.
        \item Verified property: \texttt{P=? [F (o=2) ]}
        \item Parameters: \texttt{CMAX} (positive integer): maximum counter value.
        \item Remarks: This is a random instance of the original template benchmark from \cite{paynt}.
    \end{itemize}

    \item[\benchmark{hospital}] Hospital inventory management
    \begin{itemize}
        \item From: \cite{hospital}
        \item Short description: \emph{``[The model represents] daily drug ordering in a ward of an Italian public hospital, where patient admission/discharge and drug consumption during the sojourn are subject to uncertainty.''}
        \item Verified property: \texttt{Pmax=? [ F s=7 ]}
        \item Parameters: n/a
        \item Remarks: We have extended the planning horizon to 6 weeks and used random probabilities for the daily drug consumption.
    \end{itemize}

    \item[\benchmark{nand}] von Neumann NAND multiplexing system 
    \begin{itemize}
        \item From: \cite{nand}, \prism\ benchmark suite, QComp~\cite{qcomp} benchmark set
        \item Short description: \emph{``The case study concerns NAND multiplexing, a technique for constructing reliable computation from unreliable devices.''}
        \item Verified property: \texttt{P=? [ F s=4 \& z/N<0.1 ]}
        \item Parameters: \texttt{N} (positive integer) number of inputs in each bundle, \texttt{K}: (positive integer) number of restorative stages       
    \end{itemize}

    \item[\benchmark{nd-nand}] MDP version of the previous benchmark
    \begin{itemize}
        \item From: \cite{nand}, \prism\ benchmark suite, QComp~\cite{qcomp} benchmark set
        \item Short description: (see above)
        \item Verified property: \texttt{P=? [ F s=4 \& z/N<0.1 ]}
        \item Parameters: \texttt{N} (positive integer) number of inputs in each bundle, \texttt{K}: (positive integer) number of restorative stages
        \item Remark: In the original \prism\ program, we have replaced the command
        \begin{align*}
            &\texttt{[] s=2 \& u>1 \& zy<(N-c) \& zy>0  ->}\\
            &\qquad\qquad\qquad\texttt{ p1 : <choice A> + p2: <choice B>} 
        \end{align*}
        by two commands to resolve the above probabilistic choice in a non-deterministic way.
    \end{itemize}

    \item[\benchmark{negotiation}] Alternating Offers Protocol
    \begin{itemize}
        \item From: \cite{negotiation}, \prism\ benchmark suite
        \item Short description: \emph{``This case study is about the analysis of a Negotiation Framework known as Rubinstein's Alternating Offers Protocol. In such a framework two agents, the Buyer (B) and the Seller (S), bargain over an item. ''}
        \item Verified property:\\ \texttt{P=? [F s=2 \& b=3 \& (bid=TIMELINE/2 | cbid=TIMELINE/2)]}
        \item Parameters: \texttt{TIMELINE} (positive integer).
    \end{itemize}

    \item[\benchmark{pole}] Balancing a pole
    \begin{itemize}
        \item From: \cite{paynt}
        \item Short description: Models balancing a pole in a noisy and unknown environment.
        \item Verified property: \texttt{R\{"rounds"\}=? [F x = 0 | x = MAXX ]}
        \item Parameters: \texttt{CMAX} (positive integer): maximum counter value.
        \item Remarks: This is a random instance of the original template benchmark from \cite{paynt}.
    \end{itemize}

    \item[\benchmark{tireworld}] Navigation of a vehicle
    \begin{itemize}
        \item From:  IPPC 2006 benchmark set, QComp~\cite{qcomp} benchmark set
        \item Short description: Navigation of a vehicle which can only recover from faults at specific service stations.
        \item Verified property: \texttt{Pmax=? [F var15 = 10]}
        \item Parameters: n/a
        \item Originally specified in PPDDL.
    \end{itemize}

\end{description}

\section{Experiments with fine-tuned Heuristics}
\label{app:additional-exp}

We have also encountered examples in the literature (see \Cref{table:resultsextended}) where our default heuristics does not lead to substantial reductions.
However, by increasing either the number of maximally permitted locations (as in \benchmark{tireworld}) or decreasing the allowed elimination complexity (as in \benchmark{negotitation}), we could nonetheless achieve noticeable reduction on these models, too.

Moreover, it is occasionally possible to improve performance by fine-tuning the heuristics, even if the default settings already yield good results. This is the case for, e.g., \benchmark{brp}.
\begin{table}[h]
    \centering
    \caption{
        Further experimental results with manually tuned benchmark settings.
        Recall that the default is 10, 10000.
    }
    \label{table:resultsextended}
    \begin{adjustbox}{max width=\textwidth}
        {\renewcommand{\arraystretch}{1.0}
            \setlength{\tabcolsep}{5pt}   
            \begin{tabular}{l r c  r c  r  r  r  r  r  r  r  r  r r r}
                \toprule
                \multirow{2}{*}{Name} &  \multirow{2}{*}{Type} & Prop. & Red. & Heuristics & \multirow{2}{*}{Params.} & \multicolumn{2}{c}{States} & \multicolumn{2}{c}{Transitions} & \multicolumn{2}{c}{Build time [ms]} & \multicolumn{2}{c}{Check time [ms]} & \multicolumn{2}{c}{Total time [ms]}  \\
                & & type & time & $L\,,\,T$ & &  orig. & red. & orig. & red. & orig. & red. & orig. & red. & orig. & red.  \\ \midrule
                \multirow{4}{*}{\benchmark{brp}} & \multirow{4}{*}{dtmc} & \multirow{4}{*}{\benchmark{P}} & \multirow{4}{*}{\benchmark{134}} & \multirow{4}{*}{\benchmark{10,10000}} & $2^{10}$/5 & 78.9\texttt{K} & -44\% & 106\texttt{K} & -33\% & 261 & -33\% & 22 & -38\% & \multirow{4}{*}{\benchmark{16,418}} & \multirow{4}{*}{\benchmark{-46\%}} \\
                & & & &  & $2^{11}$/10 & 291\texttt{K} & -45\% & 397\texttt{K} & -33\% & 1,027 & -39\% & 101 & -46\% & & \\
                & & & &  & $2^{12}$/20 & 1.11\texttt{M} & -46\% & 1.53\texttt{M} & -33\% & 3,945 & -48\% & 462 & -48\% & & \\
                & & & &  & $2^{13}$/25 & 2.76\texttt{M} & -46\% & 3.8\texttt{M} & -33\% & 9,413 & -47\% & 1,187 & -47\% & & \\ \midrule
                \multirow{4}{*}{\benchmark{brp}} & \multirow{4}{*}{dtmc} & \multirow{4}{*}{\benchmark{P}} & \multirow{4}{*}{\benchmark{705}} & \multirow{4}{*}{\benchmark{50,2500}} & $2^{10}$/5 & 78.9\texttt{K} & -75\% & 106\texttt{K} & -43\% & 270 & -39\% & 23 & -68\% & \multirow{4}{*}{\benchmark{16,590}} & \multirow{4}{*}{\benchmark{-55\%}} \\
                & & & &  & $2^{11}$/10 & 291\texttt{K} & -76\% & 397\texttt{K} & -43\% & 969 & -55\% & 99 & -72\% & & \\
                & & & &  & $2^{12}$/20 & 1.11\texttt{M} & -76\% & 1.53\texttt{M} & -42\% & 4,003 & -58\% & 450 & -72\% & & \\
                & & & &  & $2^{13}$/25 & 2.76\texttt{M} & -76\% & 3.8\texttt{M} & -42\% & 9,636 & -59\% & 1,140 & -71\% & & \\ \midrule
                \multirow{2}{*}{\benchmark{negotiation}} & \multirow{2}{*}{dtmc} & \multirow{2}{*}{\benchmark{P}} & \multirow{2}{*}{\benchmark{148}} & \multirow{2}{*}{\benchmark{10,1000}} & $10^4$ & 129\texttt{K} & -32\% & 184\texttt{K} & -26\% & 481 & -39\% & 22 & -49\% & \multirow{2}{*}{\benchmark{5,631}} & \multirow{2}{*}{\benchmark{-39\%}} \\
                & & & &  & $10^5$ & 1.29\texttt{M} & -32\% & 1.84\texttt{M} & -26\% & 4,930 & -43\% & 197 & -30\% & & \\ \midrule
                \multirow{1}{*}{\benchmark{tireworld}} & \multirow{1}{*}{mdp} & \multirow{1}{*}{\benchmark{P}} & \multirow{1}{*}{\benchmark{134}} & \multirow{1}{*}{\benchmark{55,10000}} & n/a & 197\texttt{K} & -24\% & 851\texttt{K} & +10\% & 923 & -21\% & 412 & -47\% & \multirow{1}{*}{\benchmark{1,335}} & \multirow{1}{*}{\benchmark{-19\%}} \\
                \bottomrule
            \end{tabular}
        }
    \end{adjustbox}
\end{table}
}

\end{document}